\documentclass[11pt,a4paper]{article}
\usepackage{amsmath,amssymb,amsfonts,amsthm}
\usepackage[pdftex]{graphicx, color}

\newcommand{\E}{{\bf{E}}}
\newcommand{\PP}{{\bf{P}}}
\newcommand{\PPP}{{\text{P}}}
\newcommand{\Var}{{\bf{Var}}}

\newtheorem{tm}{Theorem}

\newtheorem{lem}{Lemma}

\newtheorem{rem}{Remark}

\begin{document}

\bibliographystyle{plain}
\parindent=0pt
\centerline{\LARGE \bfseries Clustering coefficient }
\centerline{\LARGE \bfseries of random intersection graphs} 
\centerline{\LARGE \bfseries with infinite degree variance}
%\centerline{\LARGE \bfseries  of a random graph}

\par\vskip 3.5em

\centerline{Mindaugas  Bloznelis and Valentas Kurauskas}

\par\vskip 1.7em

%\centerline{Faculty of mathematics and informatics,}
\centerline{ Vilnius university, 
%03225 Vilnius, 
Lithuania}

\bigskip
\centerline{{\it Dedicated to Professor Micha\l\ Karo\'nski on the occasion of his 70th birthday}}

%\centerline{2016-08-06}

%\footnotetext[*]{Faculty of Mathematics and Informatics, Vilnius University, 03225 Vilnius, Lithuania}
%\centerline{ Mindaugas Bloznelis\footnote{Research supported by the Research Council of Lithuania Grant MIP-053}}

\bigskip

% Faculty of Mathematics and Informatics, Vilnius University,

% Naugarduko 24, LT-03225 Vilnius, Lithuania

%E-mail:   mindaugas.bloznelis$@$mif.vu.lt

%http://www.mif.vu.lt/$\sim$bloznelis
%\par\vskip 3.5em

\begin{abstract}
For 
a random intersection graph with a power law degree sequence 
having a finite mean and an infinite variance
we show that the global clustering
coefficient admits a tunable asymptotic distribution.
\par
\end{abstract}

\smallskip
{\bfseries key words:}  
clustering coefficient, power law, infinite variance, random intersection graph, affiliation network.

\par\vskip 2.5em

%2000 Mathematics Subject Classifications: 91D30,   05C80,  05C07, 91C20

%\par\vskip 2.5em

\section{Introduction}
The global clustering coefficient $C_G$ of a graph $G$ is 
the ratio
 $C_G=3\Delta/\Lambda$, where $\Delta$ is the number of triangles  
 and $\Lambda$ is the number of paths of length $2$.
Another way to represent the global clustering coefficient  is by
the conditional probability that
a randomly chosen  triple of vertices 
makes up a triangle given that the first two vertices are adjacent to the 
third one.
Formally,
\begin{displaymath}
C_G=\PP^*(v_1^*\sim v_2^*\,| v_1^*\sim v_3^*,\, v_2^*\sim v_3^*),
\end{displaymath}
where $(v_1^*,v_2^*,v_3^*)$ 
is an ordered triple of vertices sampled uniformly at random and the probability
$\PP^*$ refers to the sampling. By $\sim$ we denote the adjacency relation.

In this paper we study the relation between the clustering
  coefficient and the tail of the degree sequence in large complex networks. 
 We focus on  random intersection graph models 
 of real affiliation networks
 (mode two networks),
\cite{karonski1999}, \cite{godehardt2003}, \cite{BloznelisGJKR2015-1}.
They
admit tunable degree distribution 
and  
non-vanishing
clustering coefficient 
\cite{Newman2001}, \cite{Deijfen}, \cite{Bloznelis2013},
 \cite{BloznelisGJKR2015-2}.
Definition of a random intersection graph is recalled below
 in this section. 

The global clustering coefficient
$C_G$ of a realised instance $G$ of a random graph is a random variable.
We note that generally this random variable behaves differently 
depending on whether
the degree variance is finite or infinite \cite{Bloznelis2013},  \cite{JacobMorters2013},
\cite{OstroumovaSamosvat2014}.
When the degree variance is finite
the global clustering coefficient
$C_G$ can be approximated by the corresponding numerical characteristic
of the underlying random intersection graph model, the conditional probability
$\alpha_C:=\PP(v_1^*\sim v_2^*\,| v_1^*\sim v_3^*,\, v_2^*\sim v_3^*)$,
\cite{kurauskas2015}.	
 Here and below $\PP$ refers to
 all the sources of randomness defining the events considered
(these are the uniform sampling of vertices 
$(v_1^*,v_2^*,v_3^*)$ and random graph generation mechanism 
in the present context).
We remark that   $\alpha_C$ admits a simple asymptotic expression in terms
of
 the first and second moment of the degree sequence
\cite{Bloznelis2013}, \cite{GJR2012},  \cite{BloznelisGJKR2015-2},
\cite{BloznelisKurauskas2015}.

The question about the behaviour of 
the clustering coefficient $C_G$ 
when the degree variance is infinite  
remained open.
We address this question in the present paper.
Our study is analytical. 
For an infinite degree variance we show that
$C_G$ admits 
a non-degenerate asymptotic distribution 
with 
tunable characteristics in the case where the weights defining 
the underlying  random intersection graph  achieve a certain balance. 
In this way our theoretical findings contribute 
to the discussion about whether and when
a power law network 
model 
with an infinite degree variance
can 
have 
a
non-vanishing global clustering coefficient, 
cf. \cite{OstroumovaSamosvat2014}, where a negative result was obtained.

The paper is organized as follows. In this 
section we introduce
random intersection graphs,
formulate and discuss
our results.
Proofs are given in section 2.
Technical lemmas are postponed to Section 3.

\subsection{Random intersection graphs}

Random intersection graphs model  social networks, where the actors
establish communication links provided that they share
some common  attributes (collaboration networks, actor networks, etc.).
A random intersection graph $G$ on the vertex set 
$V=\{v_1,\dots, v_n\}$ is defined by a random
bipartite graph, denoted by $H$, with the bipartition $V\cup W$, 
where $W=\{w_1,\dots, w_m\}$ is an auxiliary set of attributes. Two vertices in 
$G$ are adjacent whenever they have a common neighbour in $H$. This neighbour is 
called a witness of the adjacency relation. 

In the {\it active}  graph, 
denoted by $G(n,m,\PPP)$,
vertices $v\in V$ select their neighbourhoods
$S_v\subset W$ in $H$ independently at random according to the 
probability distribution  $\PP(S_v=A)=\PPP(|A|){\binom{m}{|A|}}^{-1}$,
$A\subset W$. Here $\PPP$ is the probability distribution modeling 
the size $|S_v|$ of the 
neighbourhood of $v$ in $H$. Given the size $|S_v|$, 
the elements of $S_v$ 
are selected
uniformly at random. Two vertices $u,v$  are adjacent in $G$
whenever the random sets $S_u$ and $S_v$ 
(called attribute sets of $u$ and $v$) intersect. 

In the {\it passive} graph,  
denoted by $G^{\star}(n,m,\PPP^{\star})$, attributes
 $w\in W$ select their neighbourhoods
$D_w\subset V$ in $H$ independently at random according to the 
probability distribution  
$\PP(D_w=A)=\PPP^{\star}(|A|){\binom{n}{|A|}}^{-1}$,
$A\subset V$. 
Two vertices $u,v$ are adjacent in $G^{\star}(n,m,\PPP^{\star})$ 
whenever $u,v\in D_w$ for some $w\in W$.

The {\it inhomogeneous}
  graph, denoted by $G(n,m,\PPP_X, \PPP_Y)$, 
interpolates between
the active and passive models. It
is defined by the random bipartite graph, where attributes $w_i\in W$ and
vertices $v_j\in V$ are assigned independent random weights $X_i$ and $Y_j$
respectively. 
The weights model the attractiveness of attributes and activity of actors.
Every pair $(w_i,v_j)\in W\times V$ is
linked in $H$ with probability
$p_{ij}=\min\{1, X_iY_j/\sqrt{mn}\}$ independently of the other pairs.
Here $X_1,\dots, X_m$ and $Y_1,\dots, Y_n$ are non-negative 
independent random variables with 
the distributions $\PPP_X$ and $\PPP_Y$ respectively.

In what follows we assume that $n/m$ is bounded and it is bounded away 
from zero as $m,n\to+\infty$, denoted by $n=\Theta(m)$. 
The rationale behind this assumption is that in the range $n=\Theta(m)$ the active, 
passive and inhomogeneous models 
admit  non-degenerate asymptotic degree distributions including power laws \cite{Bloznelis2013, BloznelisDamarackas2013, BloznelisGJKR2015-2, Deijfen}. 
More importantly, in this range
these random graph models admit tunable global clustering coefficient $C_G\approx \alpha_C$,
provided that the degree variance is finite \cite{kurauskas2015}.
Therefore it is reasonable to consider the range $n=\Theta(m)$, also when
studying the global clustering coefficient
of  a power law intersection graph with an infinite degree variance.

\subsection{Results}
Let $d(v_i)$ denote the degree of 
a vertex $v_i\in V=\{v_1,\dots, v_n\}$ 
in a random intersection graph. 
We note that the random variables
$d(v_1),\dots, d(v_n)$ are identically distributed for each particular model:
active, passive and 
inhomogeneous. When speaking about the asymptotic degree distribution below 
we think about the limit in distribution of the random variable $d(v_1)$ 
as $n,m\to+\infty$.

{\it Active graph $G=G(n,m,P)$.}  In Theorem
 \ref{T1-2} below we show that an active graph with an infinite degree
 variance has the global clustering coefficient $C_{G}\approx 0$.

\begin{tm}\label{T1-2} Let $\beta>0$. 
Let $m,n\to+\infty$. Assume that $m/n\to \beta$.
 Let $Z$ be a non-negative random variable such that $\E Z<\infty$ and 
 $\E Z^2=\infty$.
Let $P$ denote the distribution of $\min\{Z,m\}$.
The global clustering coefficient of the active random graph 
$G(n,m,P)$ 
satisfies $C_G=o_P(1)$.
\end{tm}

Under conditions of Theorem \ref{T1-2} the active graph  
has a mixed Poisson asymptotic degree distribution 
assigning probabilities $\E e^{-\lambda}\frac{\lambda^k}{k!}$
 to the integers $k=0,1,\dots$, see \cite{Bloznelis2013}. 
Here $\lambda=(\E Z)\beta^{-1} Z$ is a random variable. 
In the case where  $Z$  
has a power law  with the tail index $\alpha>1$, i.e., for some
$c_z>0$ we have
\begin{equation}\label{Zalpha}
\PP(Z>t)=c_zt^{-\alpha}+o(t^{-\alpha})
\qquad
{\text{as}}
\qquad
t\to+\infty,
\end{equation}
the asymptotic degree distribution described above is a power law 
 with the same
tail index $\alpha$. For $1<\alpha\le 2$  it has a finite first moment,
infinite variance and the clustering coefficient $C_G\approx 0$.

 \medskip

\medskip

{\it Passive graph $G^{\star}=G^{\star}(n,m,\PPP^{\star})$.} In Theorem
 \ref{T1} below we show that a passive graph with an infinite degree
 variance has the global clustering coefficient $C_{G^{\star}}\approx 1$. 
By $X$ we denote a random variable with the distribution $\PPP^{\star}$.

\begin{tm}\label{T1} Let $\beta>0$.  Let $m,n\to\infty$. 
Assume that $mn^{-1}\to \beta$ and

(i) $X$ converges in distribution to a random variable $Z$;

(ii)  $\E Z^2<\infty$ and \ $\lim_{m,n\to\infty}\E X^2=\E Z^2$;

(iii) $\E Z^3=\infty$.

Then the clustering coefficient  $C_{G^{\star}}= 1-o_P(1)$.

\end{tm}

We mention that under conditions of Theorem \ref{T1}, the degree 
$d(v_1)$ converges in distribution to the compound Poisson random variable 
$d_*=\sum_{j=1}^{\zeta} {\tilde Z}_j$, see \cite{Bloznelis2013}.
Here ${\tilde Z}_1,{\tilde Z}_2,\dots$ are independent random variables with 
the common probability distribution $\PP({\tilde Z}_1=r)=(r+1)\PP(Z=r+1)/\E Z$, 
$r=0,1,2,\dots$.
 The random variable $\zeta$ 
is independent of the sequence 
${\tilde Z}_1,{\tilde Z}_2,\dots$ and has Poisson distribution 
with mean $\E \zeta= \beta^{-1}\E Z$.
 Assuming that for some $\alpha\in(3,4)$ and $c>0$ 
\begin{equation}\label{Passive-a}
 \PP(Z=r)=cr^{-\alpha}(1+o(1)) 
\qquad
{\text{as}}
\qquad
r\to+\infty,
\end{equation}
 we obtain, by Theorem 4.30 of 
\cite{FossKZ2013}, that
\begin{displaymath}
 \PP(d_*=r)
=
\PP({\tilde Z}_1=r)(\E \zeta)(1+o(1))
= 
c'\beta^{-1}r^{1-\alpha}(1+o(1))
\qquad
{\text{as}}
\qquad
r\to+\infty,
\end{displaymath}
for some constant $c'>0$. In this case 
$G^*$ has asymptotic power law degree distribution 
with a finire first moment, infinite variance
and the clustering coefficient $C_{G^*}\approx 1$.

\medskip

{\it Inhomogeneous graph $G(n,m,\PPP_X, \PPP_Y)$.} 
In Theorem \ref{T2} below we show that
the global clustering coefficient of an inhomogeneous graph with an infinite degree variance is highly determined by the ratio of the random variables
\begin{displaymath}
S_X=\sum_{i=1}^mX_i^3
\qquad
{\text{and}}
\qquad
S_Y=\sum_{j=1}^nY_j^2.
\end{displaymath}
We denote
 $a_i=\E X_1^i$, $i=1,2$, and $b_1=\E Y_1$.

\begin{tm}\label{T2} Let $\beta>0$.  Let $m,n\to\infty$. 
Assume that $mn^{-1}\to \beta$. 
Suppose that 
$\E X_1^2<\infty$, $\E X_1^3=\infty, \E Y_1<\infty, 
\E Y_1^2=\infty$. 
Denote $\kappa=\beta^{3/2}a_2^2b_1^{-1}$. We have
$C_G=\bigl(1+\kappa S_Y/S_X\bigr)^{-1}~{+o_P(1)}$.
\end{tm}

In the case where $S_X$ and $S_Y$ grow to infinity at the same rate we can 
obtain a non-trivial limit of $C_G$. The next remark addresses the case
where the distributions of $X_1^3$ and $Y_1^2$ belong to the domain of attraction 
of stable distributions having the same characteristic exponent $\alpha\le 1$.

\begin{rem}\label{R1}
 Let  $\alpha,\beta>0$.  Let $m,n\to\infty$. 
Assume that $mn^{-1}\to \beta$.
Suppose that for some  $c_x,c_y>0$ we have
\begin{equation}\label{tailsXY}
 \PP(X_1>t)=c_xt^{-3\alpha}+o(t^{-3\alpha}),
\qquad
\PP(Y_1>t)=c_yt^{-2\alpha}+o(t^{-2\alpha})
\qquad
{\text{as}}
\qquad
t\to+\infty. 
\end{equation}

(i) For
$0<\alpha<1$ the ratio $S_Y/S_X$ {\it converges in distribution} to the random variable
$c^*Z_{\alpha}/Z'_{\alpha}$, where $Z_{\alpha},Z'_{\alpha}$ 
are independent stable random variables 
with the Laplace transform 
$\E e^{-sZ_{\alpha}}=\E e^{-sZ'_{\alpha}}=e^{-s^{\alpha}}$
and $c^*=(c_y/(c_x\beta))^{1/\alpha}$.

(ii) For $\alpha=1$ the ratio $S_Y/S_X=c_y(c_x\beta)^{-1}+o_P(1)$.
\end{rem}

Let us apply 
Theorem \ref{T2}  to power law random weights (\ref{tailsXY}). 
We observe that 
$\E X_1^2,\E Y_1<\infty$ and $\E X_1^3, \E Y_1^2=\infty$
imply $2/3<\alpha\le 1$.
  For $\alpha =1$ the result  of Theorem \ref{T2}
implies that 
$C_G\approx \bigl(1+\kappa c_y/(c_x\beta)\bigr)^{-1}$ 
is asymptotically constant. For
$2/3<\alpha<1$ it implies that  $C_G$ 
{\it converges in distribution} to the random variable 
$
\bigl(
1+\kappa(c_y/(c_x\beta))^{1/\alpha}Z_{\alpha}/Z'_{\alpha}
\bigr)^{-1}$.

\bigskip

Finally, we mention that for $m/n\to\beta\in (0,+\infty)$  and 
$2/3< \alpha\le 1$, the 
inhomogeneous graph defined by power law weights (\ref{tailsXY})
has a power law 
 asymptotic degree distribution with the tail index 
$3\alpha-1$, see \cite{BloznelisDamarackas2013}. 
In particular,  the asymptotic degree distribution has 
a finite first moment and an infinite variance.

\subsection{Discussion}

One motivation of our study was the 
recent paper \cite{OstroumovaSamosvat2014},
which claims that 
``if the degree distribution obeys the power
law with an infinite variance, then the global clustering coefficient tends
to zero with high probability as the size of a graph grows.''
This may look a bit confusing in view of the fact that some 
large social networks with quite substantial 
 global clustering coefficients are believed to have a power law degree distribution with an infinite variance.
 The present study could be viewed 
 as an attempt to resolve this seemingly contradiction with the aid of 
 a known theoretical model of an affiliation network. 
 
 We observe that random intersection graphs considered in this paper admit {\it asymptotic}
 power law degree distributions, but their degree sequence is not an iid sample from a power law. We mention that some real affiliation networks
 are believed to have a power law degree sequence, but with an exponential cutoff, \cite{Newman2001aCUTOFF}, \cite{Newman2001bCUTOFF}, \cite{Redner1998}.

 In what follows we discuss the relation between the result of 
 \cite{OstroumovaSamosvat2014} and our Theorems  \ref{T1-2}, \ref{T1}, 
 \ref{T2} in some 
 detail.  To this aim we briefly recall the 
 argument of \cite{OstroumovaSamosvat2014}.
 We call a path $x\sim y\sim z$ a cherry produced by vertex $y$. For example, 
 a vertex $v_j$ of degree $d_j=d(v_j)$ produces $\binom{d_j}{2}$ cherries.
Ostroumova and Samosvat \cite{OstroumovaSamosvat2014} observed
that cherries produced by vertices of large degrees highly outnumber
the triangles of the graph. 
Indeed, among the iid degrees  $d_1,\dots, d_n$
obeying a power law with the tail index
 $1<\alpha <2$, the largest few roughly scale as $n^{1/\alpha}$.
Consequently, the number of cherries produced by 
the largest vertices roughly scale as $n^{2/\alpha}$.
 On the other hand, the number of triangles
  incident to any  vertex 
 $v_j$ does not exceed the number of cherries $\binom{d_j}{2}$. More 
 importantly, this number is bounded by the total number of edges of 
 the graph (edges needed to close  cherries produced by $v_j$). 
But for $1<\alpha$ the 
 average 
 degree is bounded and the
 total number of edges scales as $n$. This implies that only a negligible 
 fraction $n^{1-(2/\alpha)}$ of cherries produced by the largest vertices 
 are closed. Putting things together one can show that
 $3\Delta\le c\sum_j n\wedge{\binom{d_j}{2}}$ is negligible compared to
 $\Lambda =\sum_{j}\binom{d_j}{2}$. Hence $C_G=o_P(1)$.

 In a random intersection graph $G$ the triadic closure of a cherry is 
 explained by a common attribute shared by all three vertices of the cherry (triangles whose edges are witnessed by distinct attributes are rare and can be neglected).
 We exploit this clustering mechanism while evaluating the global 
 clustering coefficient $C_G$: When counting triangles we focus on cliques 
 of  $G$ induced by the neighbourhoods $D_i=D_{w_i}\subset V$ of attributes 
 $w_i\in W$ in the underlying bipartite graph $H$. Every 
 set
 $D_i$ of size ${\tilde X}_i:=|D_i|$ covers $\binom{{\tilde X}_i}{3}$ 
 triangles of $G$ and 
 the total number of triangles
obtained in this way scales as 
${\tilde S}_X=\sum_{i}\binom{{\tilde X}_i}{3}$ (overlaps can be neglected). 
In fact, this number dominates the 
total number of triangles in each of random intersection graphs considered 
in Theorems \ref{T1-2}--\ref{T2}.

In the active graph (with bounded average degree) the random variables 
${\tilde X}_i$ have the same asymptotic Poisson distribution. Hence 
${\tilde S}_X$ scales as $m$.
Furthermore, the degrees $\{d_j\}$ of  vertices $\{v_j\}$ can be 
approximated by asymptotically independent Poisson random variables having 
means $\lambda_j=Z_j\beta^{-1}\E Z_j$. Here $Z_1,\dots, Z_n$ are iid  
copies of  $Z$. Hence $\Lambda=\sum_j\binom{d_j}{2}$ 
scales as $\Theta(S_Z)$, where 
$S_Z=\sum_j Z_j^2$. 
For $\E Z^2=\infty$ the sum $S_Z$ is super-linear in $n$ and for $n=\Theta(m)$
we obtain ${\tilde S}_X/S_Z=o_P(1)$. Thus $C_G=o_P(1)$. We note that 
similarly to the case of iid degrees considered in 
\cite{OstroumovaSamosvat2014}
the number of cherries 
of active intersection graph 
scales as a sum of iid random variables having an infinite mean. 
One difference from  \cite{OstroumovaSamosvat2014} is that in our
Theorem \ref{T1-2} we have relaxed the structural 
"power law degree" condition of \cite{OstroumovaSamosvat2014}. 

The passive graph is a union of independently located cliques 
induced by the sets $D_{w_i}\subset V$, $w_i\in W$. 
Since $|D_{w_i}|={\tilde X}_i$ converges in distribution to  a random variable having infinite third moment,  we have that
%${\tilde X}_i=X_i$. For $\E X_i^3=\infty$,   
${\tilde S}_X$ is super-linear in $m$. Furthermore, we show that $\Lambda$ is dominated by the number of cherries covered by the
cliques. This number scales as $3\sum_i\binom{X_i}{3}=3{\tilde S}_X$
(we neglect overlaps again). 
Hence,  $C_G^*=1+o_P(1)$.

The inhomogeneous graph interpolates between the active and passive graphs.
The number of triangles ${\tilde S}_X$  scales as $\Theta(S_X)$ as in the 
passive graph, while  $\Lambda$ is approximately the sum 
of the number of cherries covered by large cliques (as in the passive 
graph) and  the number of cherries produced by the largest vertices (as 
in the active graph). These numbers scale as $3{\tilde S}_X$ and 
$\Theta(S_Y)$ respectively. 
In this way we obtain the approximation $C_G\approx 
(1+\Theta(S_Y/S_X))^{-1}$.
 Finally, we note that   the inhomogeneous graph is a  fitness model of 
 a real 
 affiliation network, where activity of  vertices is modeled by the 
 distribution  $P_Y$ and  attractiveness of  attributes is modeled 
 by the distribution  $P_X$. We summarize the result of 
Theorem \ref{T2}  as follows:
The global clustering coefficient is non-vanishing whenever the 
attractiveness "outweighs" the activity.

\section{Proofs}

We begin by establishing some notation. 
%Then we outline the proofs. 
Detailed proofs are given afterwards.

{\it Notation.}
By $\E_{\mathbb X}$ 
and $\PP_{\mathbb X}$ (respectively ${\tilde \E}$ and ${\tilde \PP}$) we
denote the conditional expectation and conditional probability given 
${\mathbb X}=(X_1,\dots, X_m)$ 
(respectively ${\mathbb X}$ and ${\mathbb Y}=(Y_1,\dots, Y_n)$). 
We use the notation $[k]$ for the set $\{1,2,\dots, k\}$ and 
the shorthand notation
  $\sum_{\Lambda}$ for  the double sum 
$\sum_{x\in V}\sum_{\{y,z\}\subset V\setminus\{x\}}$.
Denote 
empirical means ${\hat a}_r=m^{-1}\sum_{i\in [m]}X_i^r$ and 
${\hat b}_r=n^{-1}\sum_{j\in [n]}Y_j^r$. 

Let $G$ be the intersection graph defined by a bipartite graph $H$ 
with the bipartition $V\cup W$. 
For $x,y\in V$ and $w\in W$ we denote by ${\mathbb I}_{x\sim y}$
and ${\mathbb I}_{xw}$ the indicators of the events that $x,y$ are
adjacent in $G$ and $x,w$ are adjacent in $H$. For $v=v_j\in V$ 
and $w=w_i\in W$ we write interchangeably $Y_j$ or $Y_v$ and $X_i$ or $X_w$
also  $p_{ij}$ or $p_{wv}$. 
For $v\in V$ and $w\in W$ we denote $\lambda_{vw}=Y_vX_w(nm)^{-1/2}$.

For $w\in W$, let $D_w\subset V$ denote the set of neighbours of $w$ in $H$.
Note that each $D_w$ induces a clique in 
$G$.
Given a subgraph $G'\subset G$ and a subset $W'\subset W$
we say that the collection of sets $\{D_w,\, w\in W'\}$
is a 
cover of $G'$ 
 if every edge of $G'$ is witnessed by some
$w\in W'$ and for every $w\in W'$ there is an edge in $G'$ having no other witness
from $W'$, but $w$ (any proper subset of $W'$ can't be a cover of $G'$).

A subgraph of $G$ is labeled  ``lucky'' if it has a cover consisting of a 
single set
$D_w$, for some $w\in W$.
A subgraph is labeled ``unlucky'' if it has a cover consisting of two or more sets.
We note that a subgraph can be labeled ``lucky'' and ``unlucky'' 
simultaneously. 

The numbers of lucky and unlucky triangles ($2$-paths) are denoted by $\Delta_L$ 
and $\Delta_U$ ($\Lambda_L$ and $\Lambda_U$). The number of triangles ($2$-paths)
receiving both lucky and unlucky labels is denoted $\Delta_{LU}$ ($\Lambda_{LU}$).
Clearly, we have
\begin{equation}\label{2016-04-01+}
 \Delta=\Delta_L+\Delta_U-\Delta_{LU},
\qquad
\Lambda=\Lambda_L+\Lambda_U-\Lambda_{LU}.
\end{equation}

\begin{proof}[Proof of Theorem \ref{T1-2}] 
In the proof we use some ideas of  \cite{OstroumovaSamosvat2014}.
Before the proof we collect notation and auxiliary facts.
Let $Z_1,Z_2,\dots$ 
be iid copies of $Z$. We denote by  $\E_Z$ ($\Var_Z$) the conditional 
expectation (variance) given the sequence $\{Z_i, i\ge 1\}$. Furthermore, we 
denote $z_1=\E Z$ and 
$S_Z=\sum_{i\in[n]}Z_i^2$. Given $A\subset [n]$ we denote
$S_{Z,A}=\sum_{i\in A}Z_i^2$. By 
$d_{i,A}=\sum_{j\in A\setminus\{i\}}{\mathbb I}_{v_i\sim v_j}$ we denote
the number of neighbours from the set $\{v_j, \, j\in A\}\subset V$ of a vertex
$v_i$ in the intersection graph $G$.
In the proof we use the following  inequalities 
for the intersection probability of two independent uniformly distributed random 
subsets 
${\cal S},{\cal T}\subset W$
(see, e.g., Lemma 6 of \cite{Bloznelis2013})
 \begin{equation}\label{st19}
stm^{-1}(1-st/(m-s))
\le
\PP
\Bigl(
{\cal S}\cap {\cal T}\not=\emptyset\Bigr|\, |{\cal S}|=s, |{\cal T}|=t
\Bigr)
\le
stm^{-1}.
\end{equation}
We recall that every vertex $v_i\in V=\{v_1,\dots, v_n\}$ 
is prescribed a subset
$S_i\subset W=\{w_1,\dots, w_m\}$ of size $|S_i|=\min\{m, Z_i\}$. 
Furthermore, the 
condition $\E Z<\infty$ ensures the existence of a positive sequence 
$\varepsilon_n\downarrow 0$ such that 
\begin{equation}\label{ZZ20}
\PP(\max_{i\in [n]}Z_i<n\varepsilon_n)=1-o(1), 
\end{equation}
see Lemma \ref{L8}. Note that (\ref{ZZ20}) implies 
$\PP(\max_{i\in [n]}Z_i<m)=1-o(1)$.

Now we prove the theorem.
 For this purpose we show that 
there is a constant $c^*>0$ and a sequence $\varkappa_n\downarrow 0$ 
both depending on the distribution of $Z$ and on $\beta$  such that
\begin{eqnarray}\label{LY1-20}
&&
\PP(\Lambda>c^*S_Z)
=
1-o(1),
\\
\label{LY1-20a+}
&&
\PP(\Delta\le n^{3/2}\varkappa_n)=1-o(1),
\\
\label{LY1-20a}
&&
\Delta
=
O_P\bigl(n+n^{-3}S_Z^3\bigr).
\end{eqnarray}
Let us show that (\ref{LY1-20}), (\ref{LY1-20a+}), (\ref{LY1-20a}) 
imply $C_G=o_P(1)$.
Introduce the event 
$B=\{S_Z\le n^{3/2}\sqrt{\varkappa_n}\}$ and let ${\bar B}$ denote the complement 
event.  We have  
\begin{eqnarray}
C_G=\frac{3\Delta}{\Lambda}=
\frac{3\Delta}{\Lambda}{\mathbb I}_{B}
+
\frac{3\Delta}{\Lambda}{\mathbb I}_{{\bar B}}
=
O_P\Bigl(\frac{n}{S_Z}\Bigr)+O_P(\varkappa_n)+O_P(\sqrt{\varkappa_n})
=o_P(1).
\end{eqnarray}
Here on the event $B$ we have bounded $\Delta$ using (\ref{LY1-20a}) and
 on the event ${\bar B}$ we have applied
(\ref{LY1-20a+}).
 In the final step we invoked the  bound 
$n/S_Z=o_P(1)$, which follows by  Lemma \ref{L1}.
It remains to prove (\ref{LY1-20}), (\ref{LY1-20a+}) and (\ref{LY1-20a}).

Proof of (\ref{LY1-20}). Fix $0<a<b$ such that $p:=\PP(a<Z<b)>0$. 
Define random subsets of $[n]$
\begin{displaymath}
R=\{i:\, a<Z_i<b\},
\qquad
T=\{i:\, Z_i\le \ln^2n\},
\qquad
\Theta=\{i:\, \ln^2n<Z_i\le n\varepsilon_n\}. 
\end{displaymath}
Note that for  any $i\in [n]$ and $A\subset [n]$ 
the degree $d_i$ of a vertex $v_i$ 
is larger or equal to 
$d_{i,A}$. 
Therefore, we have
\begin{equation}\label{TU19}
 \Lambda
=
\sum_{i\in [n]}\binom{d_i}{2}
\ge 
\Lambda_T+\Lambda_\Theta,
\qquad
\Lambda_T=\sum_{i\in T}\binom{d_{i,T}}{2},
\qquad
\Lambda_\Theta=\sum_{i\in \Theta}\binom{d_{i,R}}{2}.
\end{equation}
In order to prove (\ref{LY1-20}) we show below that
\begin{equation}\label{TU20}
 \Lambda_T= (1+o_P(1))2^{-1}\beta^{-2}z_1^2S_{Z,T},
\qquad
\PP
\Bigl(
\Lambda_\Theta\ge \Bigl(\frac{ap}{4\beta}\Bigr)^2S_{Z,\Theta}
\Bigr)
=1-o(1).
\end{equation}
Indeed, (\ref{TU19}), (\ref{TU20}) combined with the identity
$S_{Z,T}+S_{Z,\Theta}=S_Z$, which holds with probability $1-o(1)$
(see  (\ref{ZZ20})), imply (\ref{LY1-20}).

\smallskip

Proof of the first relation of (\ref{TU20}). In view of Lemma  \ref{L7}
it suffices to show that 
\begin{equation}\label{LY1T}
\E_{Z}\Lambda_T= (1+o_P(1))2^{-1}\beta^{-2}z_1^2S_{Z,T},
\qquad
\Var_{Z}\Lambda_T=o_P(S_{Z,T}^2).
\end{equation}
 We note that the sum 
$S_{Z,T}=\sum_{i\in[n]}Z_i^2{\mathbb I}_{Z_i<\ln^2 n}$ is superlinear in $n$ 
as $n\to+\infty$, see Lemma
\ref{L1}.
 
To prove the first relation of (\ref{LY1T}) we 
write
\begin{equation}\nonumber
\Lambda_T=
\sum_{i\in T}
\sum_{\{j,k\}\subset T\setminus\{i\}}
{\mathbb I}_{v_i\sim v_j}
{\mathbb I}_{v_i\sim v_k}
\end{equation}
and evaluate the expectation 
\begin{displaymath}
\E_{Z}\Lambda_T
=
\sum_{i\in T}
\sum_{\{j,k\}\subset T\setminus\{i\}}
{\bar p}_{ij}
{\bar p}_{ik},
\qquad
{\bar p}_{ij}:=\PP_{Z}(v_i\sim v_j)=\PP_Z(S_i\cap S_j\not=\emptyset).
\end{displaymath}
Invoking the inequalities that follow from (\ref{st19}) 
 \begin{equation}\label{st18}
Z_iZ_jm^{-1}(1-2m^{-1}\ln^4n)
\le
{\bar p}_{ij}\le
Z_iZ_jm^{-1}
\end{equation}
we obtain 
\begin{displaymath}
\E_{Z}\Lambda_T
=
\Bigl(1+O\Bigl(\frac{\ln^4n}{m}\Bigr)\Bigr)
\sum_{i\in T}
\sum_{\{j,k\}\subset T\setminus\{i\}}
\frac{Z_i^2Z_jZ_k}{m^2}
=
(1+o_P(1))
S_{Z,T}
\frac{1}{2}\frac{{\hat z}_{1,T}^2}{\beta^2}.
\end{displaymath}
Here we denote
${\hat z}_{1,T}:=n^{-1}\sum_{i\in T}Z_i$. Finally, 
the law of large numbers implies 
${\hat z}_{1,T}=z_1+o_P(1)$. 

To prove the second relation of (\ref{LY1T}) 
we write $\Lambda_T$ in the form 
$\Lambda_T=\E_{Z}\Lambda_T+L_T+Q_T$, where
\begin{eqnarray}\nonumber
L_T
&=&
\sum_{\{i,j\}\subset T}
\bigl({\mathbb I}_{v_i\sim v_j}-{\bar p}_{ij}\bigr)
\sum_{k\in T\setminus\{i,j\}}
\bigl({\bar p}_{ik}+{\bar p}_{jk}\bigr),
\\
\nonumber
Q_T
&=&
\sum_{i\in T}\sum_{\{j,k\}\subset T\setminus\{i\}}
\bigl({\mathbb I}_{v_i\sim v_j}-{\bar p}_{ij}\bigr)
\bigl({\mathbb I}_{v_i\sim v_k}-{\bar p}_{ik}\bigr).
\end{eqnarray}
We observe that $L_T$ and $Q_T$ are conditionally 
uncorrelated (given $\{Z_n\}$). Therefore
\begin{equation}\label{VARQT}
\Var_{Z}\Lambda_T
=
\Var_{Z}L_T+\Var_{Z}Q_T.
\end{equation}
We bound the summands on the right using (\ref{st18}). 
A simple  calculation shows that
\begin{eqnarray}
\nonumber
\Var_{Z}L_T
&=&
\sum_{\{i,j\}\subset T}{\bar p}_{ij}(1-{\bar p}_{ij})
\Bigl(\sum_{k\in T\setminus\{i,j\}}
\bigl({\bar p}_{ik}+{\bar p}_{jk}\bigr)\Bigr)^2
\\
\nonumber
&
\le
&
\sum_{\{i,j\}\subset T}
\frac{Z_iZ_j}{m}
\Bigl(Z_i\beta^{-1}{\hat z}_{1,T}+Z_j\beta^{-1}{\hat z}_{1,T}\Bigr)^2
\\
\nonumber
&
\le
&
2\beta^{-3}{\hat z}_{1,T}^3\sum_{i\in T}Z_i^3.
\end{eqnarray}
Now, invoking the inequality 
$\sum_{i\in T}Z_i^3\le S_{Z,T}\max_{i\in T}Z_i\le S_{Z,T}^{3/2}$
and the bound ${\hat z}_{1,T}=O_P(1)$ we obtain 
$\Var_{Z}L_T=O_P(S_{Z,T}^{3/2})=o_P(S_{Z,T}^2)$. Furthermore, we have
\begin{eqnarray}
\nonumber
\Var_{Z}Q_T
&=&
\sum_{i\in T}\sum_{\{j,k\}\subset T\setminus\{i\}}
{\bar p}_{ij}(1-{\bar p}_{ij})
{\bar p}_{ik}(1-{\bar p}_{ik})
\le
\sum_{i\in T}\sum_{\{j,k\}\subset T\setminus\{i\}}
{\bar p}_{ij}
{\bar p}_{ik}.
\end{eqnarray}
Invoking the inequality ${\bar p}_{ij}
{\bar p}_{ik}\le Z_i^2Z_jZ_km^{-2}$ (which follows from (\ref{st18})) we obtain
\begin{displaymath}
 \Var_{Z}Q_T
\le
S_{Z,T}{\hat z}_{1,T}^2\beta^{-2}=O_P(S_{Z,T})=o_P(S^2_{Z,T}).
\end{displaymath}
Finally, (\ref{VARQT}) implies
$\Var_{Z}\Lambda_T=o_P(S_{Z,T}^2)$.

\smallskip

Proof of the second relation of (\ref{TU20}).
For every $i\in \Theta$ and $j\in R$ we have, by (\ref{st19}),
\begin{displaymath}
 \PP_Z(v_i\sim v_j)\ge 0.9 am^{-1}Z_i=:q_i.
\end{displaymath}
Here $0.9$ is a lower bound for the number $1-Z_iZ_j/(m-Z_j)$ 
valid for sufficiently large $m,n$.
We note that conditionally, given $\{Z_i,\, i\ge 1\}$ and $|R|$,
the random variable $d_{i,R}$ is a sum of independent
indicators (their number is $|R|$) each having success probability at least $q_i$.
Furthermore, $|R|$ has binomial distribution with 
 mean $np$. Given $t\ge 0$ we have
\begin{equation}\label{20+tt}
 \PP(d_{i,R}\ge t)\ge \PP\bigl(d_{i,R}\ge t\bigl|\,|R|\ge np/2\bigr)-r_1
\ge 
\PP(L_i\ge t)-r_1.
\end{equation}
Here $r_1 =\PP(|R|<np/2)$ and $L$ is the sum of $n':=\lfloor np/2\rfloor$ 
independent 
indicators with the same success probability $q_i$. Chernoff's inequality implies
\begin{equation}\label{20+ttt}
 \PP(|R|<np/2)\le e^{-np/4}=O(n^{-9}),
\qquad
\PP(L< n'q_i/2)\ge e^{-n'q_i/4}=O(n^{-9}).
\end{equation}
Note that the second bound holds uniformly in $i\in \Theta$, since $Z_i\ge \ln^2n$ 
for $i\in \Theta$. Choosing $t_i=n'q_i/2$ in (\ref{20+tt}) we obtain
\begin{displaymath}
 \PP(d_{i,R}\ge t_i, \, i\in \Theta)\ge 1-O(n^{-8}).
\end{displaymath}
This bound implies the second relation of (\ref{TU20}).

\medskip

Proof of (\ref{LY1-20a}). We recall that $\max_{i\in[n]}Z_i\le m$ with probability
$1-o(1)$. Assuming that this inequality holds we prove  below that
$\E_Z\Delta\le O_P(n+n^{-3}S_Z^3)$. Now (\ref{LY1-20a}) follows by Lemma \ref{L7} (ii).

We have $\Delta\le\Delta_L+\Delta_U$, where 
the numbers $\Delta_L$ and $\Delta_U$ of lucky and unlucky triangles satisfy
\begin{eqnarray}
 \nonumber
\Delta_L
&
\le
&
\sum_{w\in W}
\sum_{\{i,j,k\}\subset [n]}
{\mathbb I}_{w\in S_i}
{\mathbb I}_{w\in S_j}
{\mathbb I}_{w\in S_k},
\\
\nonumber
\Delta_U
&
\le
&
\sum_{
\begin{subarray}{c}
w,\tau,\varkappa\in W\\
w\not=\tau\not=\varkappa
\end{subarray}
} 
\
\sum_{\{i,j,k\}\subset [n]}
{\mathbb I}_{w\in S_i}
{\mathbb I}_{w\in S_j}
{\mathbb I}_{\tau\in S_i}
{\mathbb I}_{\tau\in S_k}
{\mathbb I}_{\varkappa\in S_j}
{\mathbb I}_{\varkappa\in S_k}.
\end{eqnarray}
Invoking the identity $\PP_Z(w\in S_i)=m^{-1}Z_i$ and inequality 
$\PP_Z(w,\tau\in S_i)\le m^{-2}Z_i^2$ we obtain
\begin{eqnarray}\nonumber
 \E_Z \Delta_L
&
\le
& m^{-2}\sum_{\{i,j,k\}\subset [n]}Z_iZ_jZ_k
\le
\beta^{-2}{\hat z}_1^3 n=O_P(n),
\\
\nonumber
 \E_Z \Delta_U
&
\le
& m^{-3}\sum_{\{i,j,k\}\subset [n]}Z_i^2Z_j^2Z_k^2
\le \beta^{-3}n^{-3}S_Z^3.
\end{eqnarray}

\medskip

Proof of (\ref{LY1-20a+}). By Lemma \ref{L8}, we can find an increasing positive
function $\psi(t)\uparrow +\infty$ as $t\to+\infty$
such that $\E Z\psi(Z)<\infty$. We can assume that $\psi(t)<t^{1/4}$, for $t\ge 1$.
Denote $\delta_n=1/\psi(n^{1/4})$ and 
$\tau_n=\E Z\psi(Z){\mathbb I}_{\{Z\ge n^{1/4}\}}$.
Put $\varkappa_n=\min\{\delta_n^{1/4},\tau_n^{1/4}\}$.
Clearly, $\delta_n\downarrow 0$, $\tau_n\downarrow 0$ and $\varkappa_n\downarrow 0$.
We observe that
\begin{equation}\label{deltau}
 \E Z^2{\mathbb I}_{\{Z<\sqrt{n}\delta_n\}}\le \sqrt{n}\delta_n z_1,
\qquad
\PP(Z\ge \sqrt{n}\delta_n)
\le 
\frac{\E Z\psi(Z){\mathbb I}_{\{Z\ge \sqrt{n}\delta_n\}}}
{\sqrt{n}\delta_n\psi(\sqrt{n}\delta_n)}
\le \frac{\tau_n}{\sqrt{n}}.
\end{equation}
Now we estimate $\Delta$. 
 We observe that the number $\Delta_i$ of triangles incident 
to a
given vertex $v_i\in V$ is at most $\binom{d_i}{2}$. Furthermore, $\Delta_i$
is always less than the total number of edges in the graph, denoted by ${\cal E}$. 
Therefore, we have
\begin{equation}\label{laikas}
 3\Delta
=
\sum_{i\in [n]}\Delta_i
\le 
\sum_{i\in [n]:\, Z_i<\sqrt{n}\delta_n}\binom{d_i}{2}
+
{\cal E}\sum_{i\in [n]:\, Z_i\ge \sqrt{n}\delta_n}1=:U_1+{\cal E}U_2.
\end{equation}
We  show below that  $\E U_1=O(n^{3/2}\delta_n)$,
$\E U_2=O(\sqrt{n}\tau_n)$
 and
$\E {\cal E}=O(n)$. These bounds together with (\ref{laikas}) 
imply  (\ref{LY1-20a+}).

For 
${\cal E}=\sum_{\{u,v\}\subset V}{\mathbb I}_{u\sim v}$  we have, by (\ref{st19}),
\begin{displaymath}
 \E {\cal E}={\binom{n}{2}}\PP(v_1\sim v_2)\le {\binom{n}{2}}\E (Z_1Z_2/m)
\le
\frac{n^2}{2m}(\E Z_1)^2=O(n).
\end{displaymath}
For $U_2=\sum_{i\in [n]}{\mathbb I}_{Z_i>\sqrt{n}\delta_n}$ we have, see 
(\ref{deltau}),
\begin{equation}\nonumber
 \E U_2\le n\PP(Z_1>\sqrt{n}\delta_n)\le \sqrt{n}\tau_n.
\end{equation}
It remains to bound $\E U_1$. 
For every $i$ we have, by (\ref{st19}), 
\begin{equation}\label{d2-11+}
\E\Bigl(\binom{d_i}{2}\Bigr|Z_i\Bigr)
=
\sum_{\{k,r\}\subset [n]\setminus\{i\}}
\PP(v_k\sim v_i |Z_i)\PP(v_r\sim v_i|Z_i)
\le
{\binom{n-1}{2}}
Z_i^2(\E Z_1)^2m^{-2}.
\end{equation}
Invoking the first inequality of (\ref{deltau}) we obtain
\begin{displaymath}
\E\binom{d_i}{2}{\mathbb I}_{\{Z_i\le\sqrt{n}\delta_n\}}
\le
{\binom{n-1}{2}}\frac{z_1^2}{m^2}\E Z_i^2
{\mathbb I}_{\{Z_i\le\sqrt{n}\delta_n\}}
\le \frac{z_1^3}{2\beta^2}\sqrt{n}\delta_n.
\end{displaymath}
Finally, we have
\begin{displaymath}
 \E U_1
=\sum_{i\in [n]}\binom{d_i}{2}{\mathbb I}_{Z_i\le \sqrt{n}\delta_n}
\le 
\frac{z_1^3}{2\beta^2}n^{3/2}\delta_n.
\end{displaymath}

\end{proof} 

\medskip

\begin{proof}[Proof of Theorem \ref{T1}] In the proof we use the notation $X_i=|D_{w_i}|$, $w_i\in W$, and $S_X=\sum_{i\in [n]}X_i^3$.
 
  We firstly 
count triangles. 
For every $w\in W$ there are $N_w:=\binom{|D_w|}{3}$ lucky
triangles covered by $D_w$. 
We have, by inclusion-exclusion, that
\begin{equation}\label{DELTAN}
 N-N^*\le \Delta_L\le N,
\qquad
{\text{where}}
\qquad
N=\sum_{w\in W}\binom{|D_w|}{3},
\qquad
N^*=\sum_{\{w,\tau\}\subset W}\binom{|D_w\cap D_\tau|}{3}.
\end{equation}
Here $\binom{|D_w\cap D_\tau|}{3}$ counts triangles
covered by $D_w$ and $D_\tau$ simultaneously.
Every unlucky triangle has its edges covered by distinct sets. 
Therefore, 
$\Delta_U$ is at most the sum
\begin{displaymath}
 N^{**}
:=
\sum_{\{x,y,z\}\subset V}\sum_{1\le i\not= j\not= k\le m}
{\mathbb I}_{\{x,y\}\subset D_i}
{\mathbb I}_{\{x,z\}\subset D_j}
{\mathbb I}_{\{y,z\}\subset D_k}.
\end{displaymath}
We estimate the total number of triangles $\Delta$ from the inequalities
 $\Delta_L\le \Delta\le \Delta_L+\Delta_U$. Hence
\begin{equation}\label{Delta-1}
|\Delta-N|\le N^*+N^{**}. 
\end{equation}

We secondly count $2$-paths. We have $\Lambda=\Lambda_L+\Lambda_U-\Lambda_{LU}$,
where $\Lambda_{LU}$ is the number of paths labeled both lucky and unlucky.
For the number of lucky paths   
$\Lambda_L=3\Delta_L$, we can evaluate 
$\Lambda_L$ using (\ref{DELTAN}).
Furthermore, the number $\Lambda_U$ of unlucky paths is at most the sum
\begin{displaymath}
M^*:=\sum_{\{w,\tau\}\subset W}|D_w\cap D_\tau|\times |D_w|\times |D_\tau|.
\end{displaymath}
Here $|D_w\cap D_\tau|\times |D_w|\times |D_\tau|$ is an upper 
bound for the number of $2$-paths with the central vertex belonging to
$D_w\cap D_\tau$ and with the endpoints belonging to $D_w\setminus D_\tau$ and
$D_\tau\setminus D_w$ respectively. 
From the inequalities $\Lambda_L\le \Lambda\le \Lambda_L+\Lambda_U$ 
we obtain
\begin{equation}\label{Lambda-1}
 |\Lambda-3N|\le 3N^*+M^*.
\end{equation}

Finally, we derive the relation 
$C_{G^*}=3\Delta/\Lambda
= 1+o_P(1)$  from (\ref{Delta-1}), (\ref{Lambda-1})
and the bounds $N^*, N^{**}, M^*=o_P(N)$ shown below.

Let us bound $N^*, N^{**}, M^*$. 
We note that the sum $S_X$ is superlinear in $m$. Indeed, 
 Lemma \ref{L1} implies that
$\PP(S_X>m\phi_m)=1+o(1)$ for some $\phi_m\uparrow +\infty$.
A simple consequence of this fact is that 
$6N=(1+o_P(1))S_X$ is superlinear in $m$ as well. Furthermore, the bounds
$N^*, N^{**}, M^*=o_P(N)$
are equivelaent to the bounds
$N^*, N^{**}, M^*=o_P(S_X)$. 
In order to show these we prove that  
\begin{equation}\label{NNMM}
\E_{\mathbb X} N^*=o_P(S_X),
\qquad 
\E_{\mathbb X} N^{**}=o_P(S_X),
\qquad 
\E_{\mathbb X} M^*=o_P(S_X),
\end{equation}
and apply  Lemma \ref{L7}.
To prove the first bound of (\ref{NNMM}) we
write $\binom{|D_w\cap D_\tau|}{3}$ in the form 
\begin{displaymath}
 \binom{|D_w\cap D_\tau|}{3}
=
\sum_{\{x,y,z\}\subset V}
{\mathbb I}_{\{x,y,z\}\subset D_w}
{\mathbb I}_{\{x,y,z\}\subset D_\tau},
\end{displaymath}
evaluate the conditional expectation  
\begin{displaymath}
 \E_{\mathbb X} N^*
=
{\binom{n}{3}}
\sum_{\{i,j\}\subset [m]}{\binom{X_i}{3}}{\binom{X_j}{3}}{\binom{n}{3}}^{-2},
\end{displaymath}
and invoke (\ref{L2-2}) of Lemma \ref{L2}.
To prove the second bound of (\ref{NNMM}) we
evaluate 
\begin{displaymath}
 \E_{\mathbb X} N^{**}
=
{\binom{n}{3}}
\sum_{1\le i \not= j\not=k\le m}
{\binom{X_i}{2}}{\binom{X_j}{2}}{\binom{X_k}{2}}{\binom{n}{2}}^{-3}
\end{displaymath}
and invoke (\ref{L3-1}) Lemma \ref{L2}.
To prove the third bound of (\ref{NNMM}) 
we evaluate 
\begin{displaymath}
 \E_{\mathbb X} M^{*}
=
\E_{\mathbb X}
\sum_{\{i,j\}\subset[m]}
X_iX_j\sum_{x\in V}{\mathbb I}_{x\in D_i}{\mathbb I}_{x\in D_j}
=n^{-1}\sum_{\{i,j\}\subset[m]}
X_i^2X_j^2
\end{displaymath}
and invoke (\ref{L4-1}) of Lemma \ref{L2}.
\end{proof}

\medskip

\begin{proof}[Proof of Theorem \ref{T2}] Before the proof we introduce 
some notation.
We fix positive sequences $\varepsilon\downarrow 0$ and $t_n\uparrow +\infty$ such 
that $\PP\bigl(\max_{i\in[n]}Y_i<\varepsilon_nt_n^{-1}n\bigr)=1-o(1)$, 
see Lemma \ref{L8}. Note that $\E Y_1{\mathbb I}_{Y_1\ge t_n}=o(1)$ implies
$n^{-1}\sum_{i\in [n]}Y_i{\mathbb I}_{Y_i\ge t_n}=o_P(1)$.
We recall that the inhomogeneous  graph $G$ is defined by
a  bipartite graph $H$ with the bipartition $V\cup W$. 
We color vertices in $V$
white and those in $W$ black. Given a bipartite graph $H'=(V',W';E')$ 
with the bipartition 
$V'\cup W'$ and the edge set $E'$, we color vertices in $V'$ white and those in 
$W'$ black.
%We denote by $Emb(H',H)$ the set of 
%embeddings of $H'$ into 
%$H$ that preserve the vertex color. (An embedding is 
%an injective mapping 
%$\phi:V'\cup W'\to V\cup W$ such that  
%$\{v,w\}\in E'$ impplies that $\phi(v)$ and $\phi(w)$ 
%are neighbours in $H$. 
Define the bipartite graphs
\begin{eqnarray}
\nonumber
&&
 H_1
=
\Bigl(
\{1,2,3\},
\{a\};
\bigl\{
\{1,a\}, \{2,a\},  \{3,a\}
\bigr\}
\Bigr),
\\
\nonumber
&&
 H_2
=
\Bigl(
\{1,2,3\},
\{a,b\};
\bigl\{
\{1,a\}, \{2,a\}, \{2,b\}, \{3,b\}
\bigr\}
\Bigr),
\\
\nonumber
&&
 H_3
=
\Bigl(
\{1,2,3\},
\{a,b,c\};
\bigl\{
\{1,a\}, \{2,a\}, \{2,b\}, \{3,b\}, \{1,c\}, \{3,c\}
\bigr\}
\Bigr),
\\
\nonumber
&&
 H_4
=
\Bigl(
\{1,2,3\},
\{a,b,c\};
\bigl\{
\{1,a\}, \{2,a\}, \{2,b\}, \{3,b\}, \{1,c\}, \{2,c\}
\bigr\}
\Bigr),
\\
\nonumber
&&
 H_5
=
\Bigl(
\{1,2,3\},
\{a,b\};
\bigl\{
\{1,a\}, \{2,a\}, \{3,a\}, \{1,b\}, \{2,b\}
\bigr\}
\Bigr).
\end{eqnarray}
For $1\le i\le 5$ we denote by ${\cal H}_i$ the set of 
copies of $H_i$ in $H$. The number of copies is denoted 
$N_i=|{\cal H}_i|$. 
We note that every $H'\in {\cal H}_1$ defines a lucky 
triangle in $G$, 
$H''\in{\cal H}_2$ defines an unlucky path in $G$, and
$H'''\in {\cal H}_3$ defines  an unlucky triangle in 
$G$.  In particular, we have $\Delta_L\le N_1$, $\Lambda_L\le 3N_1$,
$\Delta_U\le N_3$, and $\Lambda_U\le N_2$.
We call an edge $v_i\sim v_j$ of $G$ heavy 
if $Y_iY_j>\varepsilon_n n$. 
A subgraph of $G$ is called heavy if it contains
a heavy edge. Otherwise it is called light.
The number of heavy (light) copies of $H_i$ is denoted $N_i^+$ ($N_i^-$).
\medskip

 The theorem follows from (\ref{2016-04-01+}) and the relations 
\begin{eqnarray}
\label{N-N}
 N_1
&=&
(1+o_P(1)){\tilde \E}N_1=6^{-1}\beta^{-3/2}b_1^3S_X+o_P(S_X),
 \\
\label{LU-L}
N_2
&=&
(1+o_P(1)){\tilde \E}N_2=2^{-1}a_2^2b_1^2S_Y+o_P(S_Y),
\\
\label{2016-04-01-1}
&&
 \Delta_L=N_1+o_P(S_X),
\\
\label{2016-04-01-2}
&&
\Delta_U=o_P(S_Y),
\\
\label{2016-04-01-3}
&&
\Lambda_U=N_2+o_P(S_Y)+o_P(S_X),
\\
\label{2016-04-01-4}
&&
\Lambda_{LU}=o_P(S_X)+o_P(S_Y).
\end{eqnarray}
Relations (\ref{N-N}), (\ref{LU-L}) follow from 
Lemmas \ref{L7}, \ref{L5}, \ref{unlucky+}. 
It remains to 
prove 
(\ref{2016-04-01-1}-\ref{2016-04-01-4}).

We begin with establishing auxiliary facts. Denote
\begin{displaymath}
L_n
:=
n^{-1}\sum_{\{i,j\}\subset [n]}
Y_iY_j(Y_i+Y_j){\mathbb I}_{Y_iY_j>\varepsilon_n n},
\qquad
L'_n
:=
n^{-1}\sum_{\{i,j\}\subset [n]}
Y_iY_j{\mathbb I}_{Y_iY_j>\varepsilon_n n}.
\end{displaymath}
We have
\begin{eqnarray}\label{2016-04-01+21}
&&
L_n=o_P(S_Y),
\qquad
L'_n=o_P(n)
\\
\label{2016-04-05+1}
&&
N_1^+=o_P(S_X),
\qquad
N_2^+=o_P(S_Y),
\\
\label{2016-04-05+2}
&&
N_3^-=o_P(S_Y),
\qquad
N_4^-=o_P(S_Y),
\qquad
N_5^-=o_P(S_X).
\end{eqnarray}
Proof of (\ref{2016-04-01+21}). 
On the event
$\{\max_{i\in[n]}Y_i\le \varepsilon_nt_n^{-1}n\}$ which 
has probability $1-o(1)$ we have
\begin{displaymath}
 Y_iY_j(Y_i+Y_j){\mathbb I}_{Y_iY_j>\varepsilon_n n}
\le
Y_i^2Y_j
{\mathbb I}_{Y_j>t_n}+Y_j^2Y_i{\mathbb I}_{Y_i>t_n}.
\end{displaymath}
Hence $L_n
 \le 
 S_Yn^{-1}\sum_{i\in [n]}Y_i{\mathbb I}_{Y_i>t_n}$. 
 The 
 bound 
 $n^{-1}\sum_{i\in [n]}Y_i{\mathbb I}_{Y_i>t_n}=o_P(1)$
 implies the first bound of (\ref{2016-04-01+21}). The 
 second bound is obtained in a similar way.

Proof of (\ref{2016-04-05+1}).  We combine 
 Lemma \ref{L8} with 
 the
inequalities
\begin{eqnarray}\nonumber
 {\tilde\E}N_1^+
&
\le
&
 \sum_{\{x,y,z\}\subset V}\sum_{w\in W}
 \frac{Y_xY_yY_zX_w^3}{(nm)^{3/2}}
\bigl(
{\mathbb I}_{\{Y_xY_y>\varepsilon_n n}
+
{\mathbb I}_{\{Y_xY_z>\varepsilon_n n}
+
{\mathbb I}_{\{Y_yY_z>\varepsilon_n n}
\bigr)
\\
\label{2014-04-01+31}
&
\le
&
3\beta^{-3/2}{\hat b}_1n^{-1}L'_nS_X
=
o_P(S_X),
\end{eqnarray}
\begin{eqnarray}\nonumber
 {\tilde \E}N_2^+
&\le& 
\sum_{\Lambda}
\
\
\sum_{w,\tau\in W:\, w\not=\tau}
\frac{Y_yY_x^2Y_zX_w^2X_\tau^2}{n^2m^2}
\bigl({\mathbb I}_{Y_xY_y>\varepsilon_n n}+{\mathbb I}_{Y_xY_z>\varepsilon_n n}\bigr)
\\
\nonumber
&=&
\sum_{x\in V}
\ 
\
\sum_{y,z\in V\setminus\{x\}: \, y\not= z}
\
\
\sum_{w,\tau\in W:\, w\not=\tau}
\frac{Y_yY_x^2Y_zX_w^2X_\tau^2}{n^2m^2}
{\mathbb I}_{Y_xY_y>\varepsilon_n n}
\\
\label{2014-04-01+32}
&\le &
{\hat a}_2^2
{\hat b}_1L_n
=o_P(S_Y).
\end{eqnarray}
In the last steps of (\ref{2014-04-01+31}) and (\ref{2014-04-01+32}) we have used
(\ref{2016-04-01+21}).

Proof of (\ref{2016-04-05+2}).
We combine 
 Lemma \ref{L8} with 
 the
inequalities
\begin{eqnarray}\nonumber
 {\tilde \E}N_3^-
&
\le
& 
3!\sum_{\{x,y,z\}\subset V}
\
\
\sum_{\{w,\tau,\eta\}\subset W}
\frac{Y_x^2Y_y^2Y_z^2X_w^2X_\tau^2X_\eta^2}{(nm)^{3}}
{\mathbb I}_{Y_xY_y\le \varepsilon_n n}
\\
\nonumber
&
\le
&
 \varepsilon_n
 {\hat a}_2^3
 \sum_{x,y,z\in V:\, x\not=y\not=z}
\frac{Y_xY_yY_z^2}
{n^2}
\\
\nonumber
&
\le
&
 \varepsilon_n
 {\hat a}_2^3{\hat b}^2 S_Y=o_P(S_Y),
\end{eqnarray}
\begin{eqnarray}\nonumber
 {\tilde \E}N_4^-
&
\le
& 
\sum_{x,y,z\in V:\, x\not=y\not=z}
\
\
\sum_{w,\tau,\eta\in W:\, w\not=\tau\not=\eta}
\frac{Y_x^2Y_y^3Y_zX_w^2X_\tau^2X_\eta^2}{(nm)^3}
{\mathbb I}_{Y_xY_y\le \varepsilon_n n}
\\
\nonumber
&
\le
&
\varepsilon_n{\hat a}_2^3 
\sum_{x,y,z\in V:\, x\not=y\not=z}
\frac{Y_xY_y^2Y_z}{n^2}
\\
\nonumber
&
\le
&
\varepsilon_n
{\hat a}_2^3{\hat b}_1^2S_Y
=o_P(S_Y),
\end{eqnarray}
\begin{eqnarray}\nonumber
 {\tilde \E}N_5^-
&
\le
& 
\sum_{x,y,z\in V:\, x\not=y\not=z}
\
\
\sum_{w,\tau\in W:\, w\not=\tau}
\frac{Y_x^2Y_y^2Y_zX_w^3X_\tau^2}{(nm)^{5/2}}
{\mathbb I}_{Y_xY_y\le \varepsilon_n n}
\\
\nonumber
&
\le
&
\varepsilon_n{\hat a}_2S_X 
\sum_{x,y,z\in V:\, x\not=y\not=z}
\frac{Y_xY_yY_z}{\beta^{3/2}n^3}
\\
\nonumber
&
\le
&
\varepsilon_n{\hat a}_2{\hat b}_1^3\beta^{-3/2}S_X
=o_P(S_X).
\end{eqnarray}

Now we are ready to prove (\ref{2016-04-01-1}-\ref{2016-04-01-4}).

Proof of (\ref{2016-04-01-3}).
Given a light unlucky path $x\sim y\sim z$ of $G$, let
${\cal H}^{x,y,z}_{2}\subset {\cal H}_2$ denote the
set of copies of $H_2$ defining this path. Fix an
element  $H_2^*\in {\cal H}^{x,y,z}_{2}$.
% and let 
%$\{x,y,z\}\subset V$ and $\{w,\tau\}\subset W$ denote 
%its bipartition. 
All the other elements of 
${\cal H}^{x,y,z}_2$ are called duplicates. We do this for each light unlucky path.
We claim that the total number of duplicates is at 
most $N_3^{-}+N_5^{-}$. Indeed, given 
$H_2^*\in {\cal H}^{x,y,z}_{2}$ with bipartition 
denoted by 
$V'=\{x,y,z\}\subset V$ and $W'=\{w,\tau\}\subset W$,
one potential duplicate is the distinct element of ${\cal H}^{x,y,z}_2$ with the same attribute set 
$W'$. 
The union of both copies of $H_2$ defines  the complete 
bipartite graph on $V'\cup W'$ and hence a copy of 
$H_5$
on $V'\cup W'$. The duplicates of this kind are counted by $N_5^{-}$.
Remaining possible duplicates of $H_2^*$ have attribute sets different from $W'$.
We note that a duplicate $H''_2$ whose attribute set $W''\not= W'$ defines a copy of $H_3$. 
Indeed, for $W''\cap W'=\{w\}$ the
union $H_2^*\cup H''_2$ is a copy of $H_3$. 
Furthermore, for 
$W''\cap W'=\emptyset$, the
union $H_2^*\cup H''_2$ with deleted vertex $w$ is a 
copy of $H_3$. Note that distinct duplicates $H''_2$ 
define distinct copies of $H_3$. Hence, their total 
number is 
at most $N_3^{-}$. Our claim is established. It implies that 
 the number of light unlucky paths is at least 
$N_2^{-}-N_3^{-}-N_5^{-}$.  Hence the total number of unlucky paths 
\begin{displaymath}
\Lambda_U
\ge 
N_2^{-}-N_3^{-}-N_5^{-}
=
N_2-N_2^{+}-N_3^{-}-N_5^{-}.
\end{displaymath}
These inequalities in combination with (\ref{2016-04-05+1}), (\ref{2016-04-05+2})
and the simple inequality $\Lambda_U\le N_2$ imply 
(\ref{2016-04-01-3}).

Proof of (\ref{2016-04-01-4}).
A light path $x\sim y\sim z$ receives both labels lucky 
and unlucky whenever $H$ has a light copy of $H_1$ with 
the vertex set $\{x,y,z\}\cup \{w\}$ and it has a light
copy of $H_2$ with the vertex 
set $\{x,y,z\}\cup \{w',\tau\}$. Here $w$ and $w'\not=\tau$ are
arbitrary
 elements of $W$ not necessarily all distinct. The union of these two copies contains a light copy of $H_5$. Hence the number of light paths which are both lucky and unlucky is at most $N_5^{-}$.
The number of heavy unlucky paths is at most $N_2^+$.
Putting things together we obtain
$\Lambda_{LU}\le N_5^{-}+N_2^{+}$. Now 
(\ref{2016-04-05+1}), (\ref{2016-04-05+2}) imply
(\ref{2016-04-01-4}).

Proof of (\ref{2016-04-01-2}).
Every heavy unlucky triangle contains at least two heavy unlucky paths. Hence the number of  such triangles is at most $N_2^+/2$.
The number of light unlucky triangles is at most $N_3^{-}$. Hence
$\Delta_U\le N_3^{-}+N_2^{+}/2=o_P(S_Y)$.

Proof of (\ref{2016-04-01-1}).
Given a light lucky triangle $x\sim y\sim z\sim x$ of $G$, let
${\cal H}^{x,y,z}_{1}\subset {\cal H}_1$ denote the
set of copies of $H_1$ defining this triangle. Fix an
element  $H_1^*\in {\cal H}^{x,y,z}_{1}$. It is
the complete bipartite graph on the bipartition 
$\{x,y,z\}\cup\{w\}$ for some $w\in W$.
All the other elements of 
${\cal H}^{x,y,z}_1$ are called duplicates. 
We claim that the total number of duplicates is at 
most $N_5^{-}$. Indeed, for any duplicate  
$H_1'\in {\cal H}^{x,y,z}_{1}$ with bipartition 
denoted by 
$\{x,y,z\}\cup\{w'\}$, the union $H_1^*\cup H_1'$ is the complete bipartite graph on $\{x,y,z\}\cup\{w,w'\}$. We remove the edge $\{z,w'\}$ and obtain a copy of $H_5$. We conclude that the number of light lucky triangles is at least $N_1^{-}-N_5^{-}$. 
Hence $\Delta_L\ge N_1^{-}-N_5^{-}=N_1-N_1^{+}-N_5^{-}$.
These inequalities in combination with
(\ref{2016-04-05+1}), (\ref{2016-04-05+2})
and the simple inequality $\Delta_L\le N_1$ imply 
(\ref{2016-04-01-1}).

In the proof we use the fact that $n=o_P(S_Y)$ and $m=o_P(S_X)$.

\end{proof}

\begin{proof}[Proof of Remark \ref{R1}]
For $\alpha<1$,  random variables $S_X(c_x\Gamma(1-\alpha)m)^{-1/\alpha}$ 
and $S_Y(c_y\Gamma(1-\alpha)n)^{-1/\alpha}$ converge in distribution to 
independent and identically distributed  $\alpha$ stable random variables, 
say $Z_1, Z_2$, having the Laplace transform $s\to\E e^{-sZ_1}=e^{-
s^{\alpha}}$, see Theorem 2 of Section 6 of Chapter XIII of 
\cite{FellerII}. Here  $\Gamma$ is Euler's Gamma function. Hence the 
statement (i).

For $\alpha=1$, there exist deterministic sequences 
$b_{m,x}=(c_x+o(1))\ln m$ and
$b_{n,y}=(c_y+o(1))\ln n$ such that the random variables 
$m^{-1}S_x - b_{m,x}$ and $n^{-1}S_Y-b_{n,y}$ converge in distribution 
to independent asymmetric stable random variables with the characteristic
 exponent $\alpha=1$, see Theorem 3 of Section 5 of Chapter XVII 
of \cite{FellerII}. Hence the statement (ii).
\end{proof}

\section{Appendix }

In Appendix A we place auxiliary lemmas. Proofs are given in Appendix B. We remark that
Lemmas \ref{L2} and  \ref{L5}, \ref{unlucky+} refer to the notation of the proofs of   Theorems \ref{T1} and \ref{T2} respectively.

\subsection{Appendix A}

\begin{lem}\label{L1}
Let $X_1,X_2,\dots$ be a sequence of non-negative random variables converging in
 distribution to a random variable $X$. 
Assume that $\E X=\infty$. Then for some positive nonrandom sequence $\{\phi_n\}$ converging to $+\infty$ we have
\begin{equation}\label{phi}
 \PP\bigl(X_{n,1}+\cdots +X_{n,n}>\phi_nn\bigr)=1-o(1).
 \end{equation}
 Here $X_{n,1},\dots, X_{n,n}$ are iid copies of $X_n$. 
\end{lem}

\begin{lem}\label{L7}
Let $\{Z_n\}$ and $\eta=\{\eta_n\}$ be sequences
of random variables defined on the same probability 
space. 
Let $\E_{\eta}$ denote the conditional expectation given $\eta$.
Assume that 
$\E_{\eta} Z_n=0$ implies
$Z_n=0$. Then 

(i) $\E_{\eta}(Z_n-\E_{\eta}Z_n)^2=o_P\bigl((\E_{\eta}Z_n)^2\bigr)$ implies
$Z_n=(1+o_P(1))\E_{\eta}Z_n$;

(ii) $Z_n=O_P(\E_{\eta}Z_n)$.
\end{lem}

\begin{lem}\label{L8} Let $t>0$. Let $Z$ be a 
 non-negative random variable with $\E Z<\infty$.

(i) There exists a positive increasing function
$\psi(\cdot)$ such $\psi(t)\uparrow +\infty$ as $t\uparrow +\infty$ and
$\E Z\psi(Z)<\infty$. Furthermore, there exists a positive decreasing 
function $\varepsilon(\cdot)$ such that
$\varepsilon(s)\downarrow 0$ as $s\uparrow +\infty$ and 
$\PP\bigl(Z>s\varepsilon(s)\bigr)=o(s^{-1})$ for $s\to+\infty$.

(ii) Let $Z_1,Z_2,\dots$ be iid copies of $Z$. Let $n\to+\infty$.
 Then $Z_1^{1+t}+\cdots+Z_n^{1+t}=o_P(n^{1+t})$. Furthermore, 
for $\varepsilon(\cdot)$
of statement (i), we have
$\PP(\max_{1\le i\le n}Z_i>n\varepsilon(n))=o(1)$. 
\end{lem}

We remark that the functions $\psi(\cdot)$, $\varepsilon(\cdot)$
depend  on the probability distribution of $Z$.

\begin{lem}\label{L2}
Let $X_1,X_2,\dots$ be a sequence of non-negative random variables converging in
 distribution to a random variable $X$. 
Assume that $\E X^3=\infty$ and 
$0<\E X^2<\infty$. Assume that $\E X_n^2<\infty$, for each $n=1,2,\dots$, and
$\lim_n\E X_n^2=\E X^2$.
Let $\{m_n, n\ge 1\}$ be an integer sequence and, for every $n$, let  
 $X_{n,1},\dots, X_{n,m_n}$ be iid copies of $X_n$.
Let $n\to+\infty$. Assume that 
 $m_n\uparrow+\infty$.
Denote $S_{X,n}=\sum_{j\in[m_n]}X^3_{n,j}$. We have
 \begin{eqnarray}\label{L2-2}
m_n^{-3}\sum_{\{j,k\}\subset [m_n]}X^3_{n,j}X^3_{n,k} 
&=&
o_P\bigl(S_{X,n}\bigr),
\\
\label{L3-1}
 m_n^{-3}\sum_{\{j,k,r\}\subset [m_n]}X^2_{n,j}X^2_{n,k}X^2_{n,r} 
 &=&
 o_P\bigl(S_{X,n}\bigr),
\\
\label{L4-1}
 m_n^{-1}\sum_{\{j,k\}\subset [m_n]}X^2_{n,j}X^2_{n,k} 
 &=&
 o_P\bigl(S_{X,n}\bigr).
\end{eqnarray}
\end{lem}

\begin{lem}\label{L5}
Assume that $\E X_1^2<\infty$ and $\E Y_1<\infty$.

(i) For $\E X_1^3=\infty$ we have 
${\tilde \E}\bigl( N_1-{\tilde \E} N_1 \bigr)^2=o_P(S_X^2)$.

(ii) We have ${\tilde \E}N_1=6^{-1}\beta^{-3/2}b_1^3 S_X+o_P(S_X)$.
\end{lem}

\begin{lem}\label{unlucky+} Assume that $\E X_1^2<\infty$, $\E Y_1<\infty$.

(i) For  $\E Y_1^2=\infty$ we have
${\tilde \E}\bigl(N_2-{\tilde \E}N_2\bigr)^2
=
o_P(S_Y^2)$.

(ii) We have  
${\tilde \E}N_2=2^{-1}a_2^2b_1^2S_Y+o_P(S_Y)$.
\end{lem}

\subsection{Appendix B}

\begin{proof}[Proof of Lemma \ref{L1}] We need some notation.
Given random variable $Z$ and sequence $Z_1,\dots, Z_N$  
of iid copies of $Z$, we  denote $S_N(Z)=Z_1+\cdots+Z_N$. 
For a constant $A>0$  we denote the truncated random variable 
$Z[A]=Z{\mathbb I}_{\{Z\le A\}}$ and $S_N(Z[A])$ 
denotes the sum of truncated iid copies of $Z$. 

Let us prove (\ref{phi}).
Choose a sequence $\{A_n\}$ of positive constants converging to $+\infty$
(slowly enough) such that
 \begin{equation}\label{condition for A}
 \Var(X[A_n])=o(n)
 \qquad
 {\text{and}}
 \qquad
 \E X_n^i[A_n]-\E X^i[A_n]=o(1),
 \quad 
 i=1,2.
 \end{equation}
In particular, we have
$\E X_n[A_n]\to+\infty$ and $\Var X_n[A_n]=o(n)$
as $n\to+\infty$.
Now Chebyshev's inequality  implies
\begin{displaymath}
\PP\bigl(S_n(X_n[A_n]) < (n/2)\E X_n[A_n]\bigr)
\le 
4n^{-1}\bigl(\E X_n[A_n]\bigr)^{-2}\Var X_n[A_n]
=
o(1).
\end{displaymath}
Hence, for $S_n(X_n)\ge S_n(X_n[A_n])$  and $\phi_n=0.5\E X_n[A_n]$
we obtain
$\PP(S_n(X_n)< \phi_n n)=o(1)$.
\end{proof}

\begin{proof}[Proof of Lemma \ref{L7}]
Let $\PP_{\eta}$ denote the conditional probability given $\eta$ and let
$z_n$ denote $\E_{\eta}Z_n$.
  We obtain (i) by Chebyshev's inequality: $\forall$
$\varepsilon>0$ 
\begin{eqnarray}
 \nonumber
\PP(|Z_n-z_n|> \varepsilon z_n)
&=&
\E\min\bigl\{1,\,\PP_{\eta}(|Z_n-z_n|>\varepsilon z_n)\bigr\}
{\mathbb I}_{\{z_n\not=0\}}
\\
\nonumber
&\le&
\E\min
\bigl\{
1,(\varepsilon z_n)^{-2}\E_{\eta}(Z_n-z_n)^2
\bigr\}
{\mathbb I}_{\{z_n\not=0\}}
\\
\nonumber
&=&
o(1).
\end{eqnarray}
In the last step we used the fact that 
$\PP(\E_{\eta}(Z_n-z_n)^2>\delta z_n^2)=o(1)$ for any $\delta>0$.

We obtain (ii) by Markov's inequality: $\forall$ $\varepsilon>0$ 
\begin{displaymath}
 \PP(Z_n>\varepsilon^{-1}z_n)
=
\E 
\Bigl(
\PP_{\eta}(Z_n>\varepsilon^{-1}z_n) {\mathbb I}_{\{z_n\not=0\}}
\Bigr)
\le
\E\bigl( \varepsilon z_n^{-1}\E_{\eta}Z_n\bigr)
{\mathbb I}_{\{z_n\not=0\}}
\le
\varepsilon.
\end{displaymath}

\end{proof}

\begin{proof}[Proof of Lemma \ref{L8}] The poof is elementary. We present it for reader's convenience.

Proof of (i). $\E Z<\infty$ implies that the function 
$\phi(t)=\E Z {\mathbb I}_{\{Z>t\}}$ is non-increasing and $\phi(t)\to 0$ as 
$t\to+\infty$. Choose an increasing positive sequence $\{s_k\}_{k\ge 1}$ such that
$s_k\uparrow +\infty$  and $\phi(s_k)\le 2^{-k}$ and $s_1\ge 1$. Put $s_0=0$.
Consider the non-decreasing function $\psi:[0,+\infty)\to [0,+\infty)$ attaining value $k$ 
on the interval 
$[s_{k-1}, s_k]$, for $k=1,2\dots$. Clearly, $\psi(t)\to +\infty$ as 
$t\uparrow+\infty$ and we have $\E Z\psi(Z)<\infty$. 
Furthermore, we can easily 
modify $\psi(\cdot)$ in order to obtain a strictly 
increasing function satisfying the requirements of  statement (i).
Now we choose 
$\varepsilon(\cdot)$ decaying slowly enough 
($\varepsilon(s)\downarrow 0$ as $s\uparrow+\infty$) so that 
$s\varepsilon(s)\to+\infty$
and
$\varepsilon(s)\psi\bigl(s\varepsilon(s)\bigr)\to+\infty$ as $s\to+\infty$.
Finally, Markov's inequality implies 
\begin{displaymath}
\PP\bigl(Z>s\varepsilon(s)\bigr)
\le
\bigl(s\varepsilon(s)\psi(s\varepsilon(s))\bigr)^{-1}
\E Z\psi(Z){\mathbb I}_{Z\psi(Z)>s\varepsilon(s)}
=
o(s^{-1}).
\end{displaymath}

Proof of (ii).
We write $Z_1^{1+t}+\cdots+Z_n^{1+t}\le AB$, where
$A:=Z_1+\cdots+Z_n$ and 
$B:=\max_{i\in[n]}Z_i^t$, and invoke the bounds 
$A=O_P(n)$
and $B=o_P(n^{t})$. 
The first bound follows by the law of large numbers. 
The second one follows by Markov's inequality and the union bound:
$\forall \delta>0$ we have
\begin{displaymath}
 \PP(B>\delta^{t} n^{t})
=
\PP\Bigl(\max_{i\in[n]}Z_i>\delta n\Bigr)
\le
n\PP(Z_1>\delta n)
\le
\delta^{-1}\E Z_1{\mathbb I}_{Z_1>\delta n}=o(1). 
\end{displaymath}
Similarly, from (i) we obtain
$\PP\bigl(\max_{i\in[n]}Z_i>n\varepsilon (n)\bigr)
\le
n\PP\bigl(Z_1>n\varepsilon (n)\bigr)=o(1)$.
\end{proof}

\begin{proof}[Proof of Lemma \ref{L2}] Proof of (\ref{L2-2}). The  event 
${\cal A}_n=\{\max_{j\in[m_n]}X_j\ge m_n^{4/7}\}$
has probability 
\begin{displaymath}
 \PP({\cal A}_n)\le m_n\PP(X_{n,j}\ge m_n^{4/7})
\le m_n^{-1/7}\E X_{n,j}^2
=o(1).
\end{displaymath}
On the complement event, the left side of (\ref{L2-2}) is less than
\begin{displaymath}
m_n^{-3}\sum_{\{j,k\}\subset [m_n]}m_n^{12/7}X^3_{n,k}
\le
m_n^{-2/7} S_{X,n}.
\end{displaymath}   

Proof of (\ref{L3-1}). Denote $S_{*n}=m_n^{-1}\sum_{i\in [m_n]}X_{n,j}^2$.
The relation 
$\E S_{*n}=\E X_{n}^2\to\E X^2$ implies $S_{*n}=O_P(1)$. The left side of 
(\ref{L3-1}) is less than $S_{*n}^3=O_P(1)$. The right side is superlinear in $m$, by 
Lemma \ref{L1}.

Proof of (\ref{L4-1}). The left side of (\ref{L4-1}) is less than
$m_nS_{*n}^2=O_P(m_n)$. The right side is superlinear in $m$, by 
Lemma \ref{L1}.
\end{proof}

\begin{proof}[Proof of Lemma \ref{L5}]
It is convenient to write $N_1$ in the form
\begin{displaymath}
 N_1=\sum_{w\in W}U_w,
\qquad
U_w
:=
\sum_{\{x,y,z\}\subset V}
{\mathbb I}_{xw}{\mathbb I}_{yw}{\mathbb I}_{zw}.
\end{displaymath}

Proof of (i).
Given ${\mathbb X}$ and ${\mathbb Y}$, the random variables $U_w$, $w\in W$ are 
 conditionally independent. Hence 
\begin{equation}\label{NU}
 {\tilde \E}(N_1-{\tilde \E}N_1)^2
=\sum_{w\in W}{\tilde \E}{\tilde U}^2_w,
\qquad
{\text{where}}
\qquad
{\tilde U}_w:=U_w-{\tilde \E}U_w.
\end{equation}
We bound every expectation  ${\tilde \E}{\tilde U}^2_w$ using conditional
Hoeffding's decomposition ${\tilde U}_w=L_w+ Q_w+K_w$, where the components
\begin{eqnarray}\label{LQK}
&&
L_w=
\sum_{x\in V}{\tilde \E}({\tilde U}_w|{\mathbb I}_{xw}),
\qquad
Q_w
=
\sum_{\{x,y\}\subset V}
{\tilde \E}\bigl({\tilde U}_w-L_w\bigr|{\mathbb I}_{xw},{\mathbb I}_{yw}\bigr),
\\
\nonumber
&&
K_w
=
\sum_{\{x,y,z\}\subset V}
{\tilde \E}\bigl({\tilde U}_w-L_w-Q_w\bigr| {\mathbb I}_{xw},
{\mathbb I}_{yw},{\mathbb I}_{zw}),
\end{eqnarray}
called the linear, quadratic and cubic part of the decomposition, are conditionally uncorrelated. 
We have in particular that
\begin{equation}\nonumber
{\tilde \E}{\tilde U}^2_w
=
{\tilde \E}L^2_w+{\tilde \E}Q^2_w+{\tilde \E}K^2_w.
\end{equation}
Moreover the summands of all three  sums of (\ref{LQK}) are conditionally 
uncorrelated (given ${\mathbb X}, {\mathbb Y}$.
Now (i) follows from (\ref{NU}) and the bounds shown below
\begin{equation}\label{LQK+}
\sum_{w\in W}{\tilde \E}L^2_w=o_P(S^2_X),
\qquad
\sum_{w\in W}{\tilde \E}Q^2_w=o_P(S^2_X),
\qquad
\sum_{w\in W}{\tilde \E}K^2_w=o_P(S^2_X).
\end{equation}
Let us prove (\ref{LQK+}). Denote, for $x,y,z\in V$ and $w\in W$,
\begin{equation}\nonumber
s_{x|w}
= 
\sum_{\{y,z\}\subset V\setminus \{x\}}p_{yw}p_{zw},
\qquad
s_{xy|w}
=
\sum_{z\in V\setminus\{x,y\}}p_{zw}.
\end{equation}
A straightforward calculation shows that
\begin{eqnarray}\nonumber
&&
{\tilde \E}({\tilde U}_w|{\mathbb I}_{xw})
=
({\mathbb I}_{xw}-p_{xw})s_{x|w}
=:
l_w(x),
\\
\nonumber
&&
{\tilde \E}\bigl({\tilde U}_w-L_w\bigr|{\mathbb I}_{xw},{\mathbb I}_{yw}\bigr)
=
({\mathbb I}_{xw}-p_{xw})({\mathbb I}_{yw}-p_{yw})s_{xy|w}
=:
q_w(x,y),
\\
\nonumber
&&
{\tilde \E}\bigl({\tilde U}_w-L_w-Q_w\bigr| {\mathbb I}_{xw},
{\mathbb I}_{yw},{\mathbb I}_{zw})
=
({\mathbb I}_{xw}-p_{xw})({\mathbb I}_{yw}-p_{yw})({\mathbb I}_{zw}-p_{zw})
=:
k_w(x,y,z).
\end{eqnarray}
Invoking the simple inequalities 
$s_{x|w}
\le
X_w^2\beta_n^{-1}{\hat b}_1^2$ 
and
$s_{xy|w}\le X_w\beta_n^{-1/2}{\hat b}_1$ 
we obtain
\begin{eqnarray}\label{lqk}
{\tilde \E}l^2_w(x)
&=&
{\tilde \E}\bigl(({\mathbb I}_{xw}-p_{xw})s_{x|w}\bigr)^2
\le
p_{xw}s_{x|w}^2
\le
X_w^5 \beta_n^{-5/2}Y_xn^{-1}{\hat b}_1^4,
\\
\nonumber
{\tilde \E}q^2_{w}(x,y)
&=&
{\tilde \E}\bigl(({\mathbb I}_{xw}-p_{xw})({\mathbb I}_{yw}-p_{yw})s_{xy|w}\bigr)^2
\le 
p_{xw}p_{yw}s^2_{xy|w}
\le
X_w^4\beta_n^{-2}n^{-2}Y_xY_y{\hat b}_1^2,
\\
\nonumber
{\tilde \E}k^2_w(x,y,z)
&\le& 
p_{xw}p_{yw}p_{zw}
\le 
X_w^3\beta_n^{-3/2}n^{-3}Y_xY_yY_z.
\end{eqnarray}
We note that for $x,y,z\in V$ the 
random variables $l_w(x)$, $q_w(x,y)$ and $k_w(x,y,z)$ are uncorrelated. Hence
\begin{equation}\nonumber
{\tilde \E}L^2_w
=
\sum_{x\in V}{\tilde \E}l^2_w(x),
\qquad
{\tilde \E}Q^2_w
=
\sum_{\{x,y\}\subset V}{\tilde \E}q^2_w(x,y),
\qquad
{\tilde \E}K^2_w
=
\sum_{\{x,y,z\}\subset V}{\tilde \E}k^2_w(x,y,z)
\end{equation}
Now from (\ref{lqk}) we obtain the bounds
\begin{equation}\nonumber
{\tilde \E}L^2_w=O_P\Bigl(\sum_{w\in W}X_w^5\Bigr),
\qquad
{\tilde \E}Q^2_w
=
O_P\Bigl(\sum_{w\in W}X_w^4\Bigr),
\qquad
{\tilde \E}K^2_w
=
O_P\Bigl(\sum_{w\in W}X_w^3\Bigr).
\end{equation}
Next, we apply H\"older's inequality. For $r=3,4,5$, we have
\begin{displaymath}
 \Bigl(\sum_{w\in W}1\cdot X_w^r\Bigr)^{6/r}
\le
\Bigl(\sum_{w\in W}1\Bigr)^{(6-r)/r}
\sum_{w\in W}X_w^6
\le 
m^{(6-r)/r}
S_X^2 .
\end{displaymath}
Finally, 
from the bound $m=o_P(S_X)$, which holds for $\E X_1^3=\infty$, see Lemma
\ref{L1},
we obtain
\begin{displaymath}
 \sum_{w\in W}X_w^r
\le 
m^{(6-r)/6}S_X^{r/3}
=
o_P\bigl(S_X^{(6-r)/6}\bigr)S_X^{r/3}
=o_P(S_X^2).
\end{displaymath}

\medskip

Proof of (ii). Denote 
$H_w=\sum_{x\in V}\lambda_{xw}$ and $R_w=H^3_w-6{\tilde \E}U_w$.
We have
\begin{equation}\nonumber
{\tilde \E}N=\sum_{w\in W} {\tilde \E}U_w
=
6^{-1}\sum_{w\in W}H^3_w
-
6^{-1}\sum_{w\in W}R_w.
\end{equation}
A straightforward calculation shows that
\begin{displaymath}
 \sum_{w\in W}H^3_w
=
\beta_n^{-3/2}
{\hat b}_1^3
S_X
=
(1+o_P(1))
\beta^{-3/2}
b_1^3S_X. 
\end{displaymath}
Hence, it remains to prove that
 $\sum_{w\in W}R_w=o_P(S_X)$.
To show this bound we write $R_w=R_{1,w}+R_{2,w}$, where
\begin{displaymath}
R_{1,w}=H_w^3-Z_w^3,
\qquad
R_{2,w}=Z_w^3-6{\tilde \E}U_w,
\qquad
{\text{and}}
\qquad
Z_w=\sum_{x\in V}p_{xw},
\end{displaymath}
and establish the bounds
\begin{equation}\label{RR}
\sum_{w\in W}R_{1,w}=o_P(S_X)
\qquad
{\text{and}}
\qquad
\sum_{w\in W}R_{2,w}=o_P(S_X).  
\end{equation}
We first prove the second bound of (\ref{RR}). We have
\begin{eqnarray}
\label{ZU}
&&
0
\le 
R_{2,w}
=
\sum_{x\in V}p^3_{xw}
+
3\sum_{x\in V}\sum_{y\in V\setminus\{x\}}p^2_{xw}p_{yw}
\le
\beta_n^{-3/2}X_w^3\bigl(n^{-2}{\hat b}_3+3n^{-1}{\hat b}_1{\hat b}_2\bigr).
\end{eqnarray}
In the last step we used  $p_{xw}\le \lambda_{xw}$. Next, invoking
the bounds
$n^{-2}{\hat b}_1^3, n^{-1}{\hat b}_2=o_{P}(1)$, which hold for $\E Y_1<\infty$, 
by Lemma (\ref{L8}), we obtain
\begin{displaymath}
0
\le
\sum_{w\in W}R_{2,w}
\le 
\beta_n^{-3/2}S_X 
\bigl(n^{-2}{\hat b}_3+3n^{-1}{\hat b}_1{\hat b}_2\bigr)
=
o_P(S_X).
\end{displaymath}
Let us  prove the first bound of (\ref{RR}).
We note that  $\E X_1^2<\infty$,
$\E Y_1<\infty$ imply that 
\begin{equation}\label{AB}
 \E X_1^2{\mathbb I}_{\{X_1>\sqrt{m}\}}\to 0,
\qquad
\E Y_1{\mathbb I}_{\{Y_1>\sqrt{n}\}}\to 0.
\end{equation}
We select a sequence $\delta_n\downarrow 0$
such that $\E Y_1{\mathbb I}_{\{Y_1>\sqrt{n}\}}=o(\delta_n)$ and introduce events
\begin{displaymath}
{\cal A}=\Bigl\{\max_{i\in[m]}X_i\le\sqrt{m}\Bigr\},
\qquad
{\cal B}
=
\Bigl\{n^{-1}\sum_{j\in [n]}Y_j{\mathbb I}_{\{Y_j>\sqrt{n}\}}\le \delta_n\Bigr\}.
\end{displaymath}
We claim that $\PP({\cal A}), \PP({\cal B})=1-o(1)$.
Indeed, by Markov's inequality and (\ref{AB}) 
\begin{eqnarray}\nonumber
 &&
1-\PP({\cal A})
\le 
\sum_{i\in[m]}\PP(X_i>\sqrt{m})
\le 
\E X_1^2{\mathbb I}_{\{X_1>\sqrt{m}\}}
\to 0,
\\
\nonumber
&&
1-\PP({\cal B})
\le 
(n\delta_n)^{-1}\sum_{j\in[n]}\E Y_j{\mathbb I}_{\{Y_j>\sqrt{n}\}}
= 
\delta_n^{-1}\E Y_1{\mathbb I}_{\{Y_1>\sqrt{n}\}}
\to 0.
\end{eqnarray}
Assuming that  events ${\cal A}$ and ${\cal B}$ hold
 we estimate the difference
\begin{equation}\label{ZH4}
 H_w-Z_w=\sum_{x\in V}(\lambda_{xw}-1){\mathbb I}_{\{\lambda_{xw}>1\}}
\le
\sum_{x\in V}\lambda_{xw}{\mathbb I}_{\{Y_x>\sqrt{n}\}}
\le X_w\beta_{n}^{-1/2}\delta_n.
\end{equation}
Here we used the inequality 
${\mathbb I}_{\{\lambda_{xw}>1\}}\le {\mathbb I}_{\{Y_x>\sqrt{n}\}}$, which holds
for $X_w\le\sqrt{m}$. Invoking (\ref{ZH4}) in the inequalities
\begin{eqnarray}\nonumber
\label{ZH}
&&
0
\le 
R_{1,w}
= 
(H_w-Z_w)(Z_w^2+Z_wH_w+H_w^2)\le 3(H_w-Z_w)H_w^2,
\end{eqnarray}
and using the identity $H_w^2=X_w^2\beta_n^{-1}{\hat b}_1^2$,  we obtain
\begin{displaymath}
 \sum_{w\in W}R_{1,w}
\le S_X\beta_n^{-3/2}{\hat b}_1^2\delta_n.
\end{displaymath}
For the latter inequality holds with probability $1-o(1)$ and $\delta_n=o(1)$, 
we conclude that $\sum_{w\in W}R_{1,w}=o_P(S_X)$.

\end{proof}

\begin{proof}[Proof of Lemma \ref{unlucky+}]
Proof of (i). In the proof we make use of Hoeffding's decomposition.
Let ${\mathbb I}_j$, $j\in [4]$ be independent Bernoulli random variables
with positive success probabilities $p_j$, $j\in [4]$.
Hoeffding's decomposition represents the random variable 
$T={\mathbb I}_1{\mathbb I}_2{\mathbb I}_3{\mathbb I}_4-p_1p_2p_3p_4$ 
by the sum of 
uncorrelated $U$ statistics of increasing order
\begin{equation}\label{Hoeffding+}
T
=U_1+U_2+U_3+U_4,
\qquad
U_1=\sum_{j\in [4]}T_j,
\qquad
U_2=\sum_{\{i,j\}\subset [4]}T_{ij},
\qquad
U_3=\sum_{\{i,j,k\}\subset [4]}T_{ijk}.
\end{equation}
The first, second, and third order terms
$T_{i}$, $T_{ij}$, and $T_{ijk}$ are defined iteratively as follows 
\begin{eqnarray}\label{Hoeffding++}
T_i=\E(T|{\mathbb I}_i),
\qquad
T_{ij}=\E\bigl(T-U_1\bigl|{\mathbb I}_i, {\mathbb I}_j\bigr),
\qquad
T_{ijk}=\E\bigl(T-U_1-U_2\bigl|{\mathbb I}_i,{\mathbb I}_j,{\mathbb I}_k\bigr).
\end{eqnarray}
Denoting $p=p_1p_2p_3p_4$ and $p_{i}^*=p/p_i$, 
$p_{ij}^*=p/(p_ip_j)$, $p_{ijk}^*=p/(p_ip_jp_k)$ we have
\begin{eqnarray}\nonumber
&&
 T_i=({\mathbb I}_i-p_i)p_i^*,
\qquad
T_{ij}=({\mathbb I}_{i}{\mathbb I}_{j}-p_ip_j)p_{ij}^*-T_i-T_j,
\\
\nonumber
&&
T_{ijk}=({\mathbb I}_{i}{\mathbb I}_{j}{\mathbb I}_{k}-p_ip_jp_k)p^*_{ijk}-
T_{ij}-T_{ik}-T_{jk}-T_i-T_j-T_k.
\end{eqnarray}

The fourth order term $U_4=T_{1234}:=T-U_1-U_2-U_3$. We note that 
various terms of Hoeffding's decomposition are mutually 
uncorrelated. 

Let us prove the lemma. Denote 
$\Lambda_*=N_2-{\tilde \E}N_2$ and 
$T^{yxz}_{(w,\tau)}
={\mathbb I}_{yw}{\mathbb I}_{xw}
{\mathbb I}_{x\tau}{\mathbb I}_{z\tau}
-
{\tilde\E}{\mathbb I}_{yw}{\mathbb I}_{xw}
{\mathbb I}_{x\tau}{\mathbb I}_{z\tau}
$.
We have 
\begin{equation}\label{LambdaVAR}
\Lambda_*
=
\sum_{x\in V}
\,
\sum_{\{y,z\}\subset V\setminus\{x\}}
\,
\sum_{w\in W}
\sum_{\tau\in W\setminus\{w\}}T^{yxz}_{(w,\tau)}.
\end{equation}

We decompose every  $T^{yxz}_{(w,\tau)}$ using (\ref{Hoeffding+})
and invoke these decompositions in  (\ref{LambdaVAR}). We then group the 
first order 
terms, the second order terms, etc. and obtain 
Hoeffding's decomposition of $\Lambda_*$,
\begin{displaymath}
 \Lambda_*=U_1^*+U_2^*+U_3^*+U_4^*.
\end{displaymath}
We specify the linear part $U_1^*$ (the sum of the first order terms),
 quadratic part $U_2^*$ (the sum of the second order terms), etc.
 in  
(\ref{Hoeffding+++}) below.
For this purpose we introduce 
some more notation. Consider the complete bipartite graph
${\cal K}_{V,W}$ with the bipartition $V\cup W$. Let 
${\cal E}=\{(y,w): y\in V,\, w\in W\}$ denote the set of 
edges of ${\cal K}_{V,W}$.
Let ${\cal E}^{*}$
denote the set of paths of length $4$  
which start from $V$. 
After we remove an edge of such a path we obtain a triple of edges, which we call trunk.
The set of trunks is denoted ${\cal E}^{**}$. For any edge $a=(yw)\in {\cal E}$
we denote ${\mathbb I}_{a}={\mathbb I}_{yw}$ the indicator of the event that 
vertex $y$ is linked to the attribute $w$ in the random bipartite graph $H$.
We also denote $p_a={\tilde \E}{\mathbb I}_a$. Furthermore
for distinct edges $a,b,c,d\in{\cal E}$
we denote
 \begin{eqnarray}\nonumber
  t_a
&
=
&
{\mathbb I}_a-p_a,
\qquad
t_{ab}
=
\bigl({\mathbb I}_a{\mathbb I}_b-p_ap_b\bigr)
-
({\mathbb I}_a-p_a)p_b-({\mathbb I}_b-p_b)p_a,
\\
\nonumber
t_{abc}
&
=
&
\bigl({\mathbb I}_a{\mathbb I}_b{\mathbb I}_c-p_ap_bp_c\bigr)-
t_{ab}p_c-t_{ac}p_b-t_{bc}p_a-t_ap_bp_c-t_bp_ap_c-t_cp_ap_b.
 \end{eqnarray}
Finally, $t_{abcd}$ is defined as $T_{1234}$ above, but for 
$T={\mathbb I}_{a}{\mathbb I}_{b}{\mathbb I}_{c}{\mathbb I}_{d}-p_ap_bp_cp_d$.

A calculation shows that 
\begin{eqnarray}\label{Hoeffding+++}
&&
U_1^*=
\sum_{a\in{\cal E}}t_aQ_a,
\qquad\qquad
\quad
\
\
U_2^*
=
\sum_{\{a,b\}\subset{\cal E}}t_{ab}Q_{ab},
\\
\nonumber
&&
U_3^*
=
\sum_{\{a,b,c\}\in{\cal E}^{**}}t_{abc}Q_{abc},
\quad
U_4^*
=
\sum_{\{a,b,c,d\}\in {\cal E}^*}t_{abcd},
\end{eqnarray}
where coefficients $Q_a$, $Q_{ab}$ and $Q_{abc}$ are given below.
 For any $a=(y,w)$ we have
\begin{eqnarray}
&&
Q_a
=
Q_{a1}+Q_{a2},
\\
\nonumber
&&
Q_{a1}=\sum_{x\in V\setminus\{y\}}p_{xw}
\sum_{\tau\in W\setminus\{w\}}p_{x\tau}
\sum_{z\in V\setminus\{x,y\}}p_{z\tau},
\qquad
Q_{a2}=\sum_{x\in V\setminus\{y\}}p_{xw}
\sum_{\tau\in W\setminus\{w\}}p_{y\tau}
\sum_{z\in V\setminus\{x,y\}}p_{z\tau}.
\end{eqnarray}
We note that sums $Q_{a1}$ and $Q_{a2}$  
represent $4$-paths, where $y$ has degree $1$ and degree $2$ respectively
(e.g., paths $y\sim w\sim x\sim \tau\sim z$ and $x\sim w\sim y\sim \tau\sim z$).
Furthermore, for a non incident pair $a=(y,w)$ and $c=(x,\tau)$ we have
\begin{eqnarray}\label{Z21}
&&
Q_{ac}=Q_{ac1}+Q_{ac2}+Q_{ac3}
\\
\nonumber
&&
Q_{ac1}
=
p_{y\tau}
\sum_{z\in V\setminus\{y,x\}}p_{z\omega},
\qquad
Q_{ac2}
=
p_{xw}
\sum_{z\in V\setminus\{y,x\}}p_{z\tau},
\qquad
Q_{ac3}=\sum_{z\in V\setminus\{y,x\}}p_{zw}p_{z\tau},
\end{eqnarray}
The sum $Q_{ac1}$ ($Q_{ac2}$)
represents $4$-paths, where $y$ ($x$) has degree $2$
(e.g., paths $x\sim \tau\sim y\sim w\sim z$ and $y\sim w\sim x\sim \tau\sim z$). 
The sum $Q_{ac3}$ 
represents $4$-paths, where $y$ and $x$ has degree $1$ (e.g., paths 
 $y\sim w\sim z\sim \tau\sim x$).
Similarly, for incident pairs 
$a=(y,w)$, $b=(x,w)$ and  $b=(x,w)$, $c=(x,\tau)$
we have
\begin{eqnarray}\label{Z22}
 Q_{ab}
=
\sum_{z\in V\setminus\{x,y\}}\sum_{\tau\in W\setminus\{w\}}
\bigl(p_{x\tau}p_{z\tau}+p_{y\tau}p_{z\tau}\bigr)
\qquad
Q_{bc}
=
\sum_{\{y,z\}\subset V\setminus\{x\}}
\bigl(p_{yw}p_{z\tau}+p_{zw}p_{y\tau}\bigr).
\end{eqnarray}
Finally, for a trunk $\{a,b,c\}$ which makes up a $3$-path, say,
$a=(yw)$, $b=(xw)$, $c=(x\tau)$, 
we have $Q_{a,b,c}=\sum_{z\in V\setminus\{x,y\}}p_{z\tau}$. For a trunk $\{a,b,d\}$
which is not a path (a union of $2$-path and an edge), say,
$a=(yw)$, $b=(xw)$ and $d=(z\tau)$, we have $Q_{a,b,d}=p_{x\tau}$.

Now we estimate ${\tilde \E }\Lambda_*^2$. From the fundamental property of 
Hoeffding's decomposition that various terms are uncorrelated we obtain 
that
\begin{eqnarray}\nonumber
{\tilde \E }\Lambda_*^2 
&=&
\sum_{a\in{\cal E}}Q_a^2{\tilde E}t_a^2
+
\sum_{\{a,b\}\subset {\cal E}}Q_{ab}^2{\tilde E}t_{ab}^2
+
\sum_{\{a,b,c\}\in {\cal E}^{**}}Q_{abc}^2{\tilde E}t_{abc}^2
+
\sum_{\{a,b,c,d\}\subset {\cal E}^*}{\tilde E}t_{abcd}^2.
\end{eqnarray}
It remains to
show that the sums in the right, which we denote  by $Z_1,Z_2,Z_3,Z_4$, are of order
$o_P(S_Y^2)$. 
For this purpose
we combine the expressions of $Q_{a\dots c}$ obtained above with the simple
 inequalities
\begin{displaymath}
 {\tilde \E}t_a^2\le p_a,
\qquad
 {\tilde \E}t_{ab}^2\le C p_ap_b,
\qquad
 {\tilde \E}t_{abc}^2\le C p_ap_bp_c,
\qquad
 {\tilde \E}t_{abcd}^2\le C p_ap_bp_cp_d.
\end{displaymath}
Here $C$ is an absolute constant. We also use the inequalities 
$p_{xw}\le(nm)^{-1/2}Y_xX_w$.

Proof of the bound $Z_1=o_P(S_Y^2)$. We have
\begin{displaymath}
 Z_1
\le 
\sum_{y\in V}\sum_{w\in W}p_{yw}(Q_{(yw)1}+Q_{(yw)2})^2.
\end{displaymath}
Invoking the inequalities 
$(Q_{(yw)1}+Q_{(yw)2})^2\le 2Q_{(yw)1}^2+2Q_{(yw)2}^2$
and
\begin{eqnarray}\nonumber
Q_{(yw)1}
&
\le
&
 \sum_{x\in V\setminus\{y\}}
\sum_{\tau\in W\setminus\{w\}}
\sum_{z\in V\setminus\{x,y\}}
\frac{Y_xX_w}{\sqrt{nm}}
\frac{Y_xX_{\tau}}{\sqrt{nm}}
\frac{Y_zX_\tau}{\sqrt{nm}}
\le 
\frac{X_w}{\sqrt{m}}\frac{S_Y}{\sqrt{n}}{\hat a}_2{\hat b}_1,
\\
\nonumber
Q_{(yw)2}
&
\le
&
\sum_{x\in V\setminus\{y\}}
\sum_{\tau\in W\setminus\{w\}}
\sum_{z\in V\setminus\{x,y\}}
\frac{Y_xX_w}{\sqrt{nm}}
\frac{Y_yX_{\tau}}{\sqrt{nm}}
\frac{Y_zX_\tau}{\sqrt{nm}}
\le
\frac{X_w}{\sqrt{m}}Y_y\sqrt{n}{\hat a}_2{\hat b}_1^2,
\end{eqnarray}
we obtain
\begin{displaymath}
Z_1
\le 
2
{\hat a}_2^2{\hat b}_1^3
\frac{S_Y^2}{\sqrt{n}}\frac{S_X}{m^{3/2}}
+
2
{\hat a}_2^2{\hat b}_1^4
\frac{S_X}{m^{3/2}}\sqrt{n}\sum_{y\in V}Y_y^3.
\end{displaymath}
Note that $\E X_1^2<\infty$ implies $S_Xm^{-3/2}=o_P(1)$. Furthermore, we have
${\hat a}_2,{\hat b}_1=O_P(1)$.
Hence the first summand is $o_P(S_Y^2)$. 
To show that the second summand is $o_P(S_Y^2)$
we use the fact (which follows from  $\E Y_1^2=\infty$ by Lemma \ref{L1}) that
$n=o_P(S_Y)$ and invoke inequalities
 \begin{equation}\label{Y3}
 \sum_{y\in V}Y_y^3\le S_Y\max_{y\in V}Y_y\le S_Y (S_Y)^{1/2}.
\end{equation}
We obtain $\sqrt {n}\sum_{y\in V}Y_y^3=o_P(\sqrt{S_Y})S_Y^{3/2}=o_P(S_Y^2)$.
We conclude that $Z_1=o_P(S_Y^2)$.

Proof of the bound $Z_2=o_P(S_Y^2)$. We split $Z_2=Z_{21}+Z_{22}+Z_{23}$, where
the sum 
\begin{displaymath}
 Z_{21}
=
\sum_{\{x,y\}\subset V}
\sum_{w\in W}
\sum_{\tau\in W\setminus\{w\}}
t_{(yw)(x\tau)}^2Q_{(yw)(x\tau)}^2
\end{displaymath}
accounts for pairs of non incident edges $a=(yw)$ and $c=(x\tau)$, while the sums
\begin{displaymath}
 Z_{22}=\sum_{\{x,y\}\subset V}\sum_{w\in W}t_{(yw)(xw)}^2Q_{(yw)(xw)}^2
\qquad
{\text{and}}
\qquad
Z_{23}=\sum_{x\in V}\sum_{\{w,\tau\}\subset W}t_{(xw)(x\tau)}^2Q_{(xw)(x\tau)}^2
\end{displaymath}
account for pairs of incident edges
$a=(y,w)$, $b=(x,w)$ and  $b=(x,w)$, $c=(x,\tau)$ respectively.
To estimate $Z_{21}$ we use (\ref{Z21}) and obtain that 
\begin{displaymath}
 Q_{(yw)(x\tau)}
\le 
\frac{X_wX_{\tau}}{m}
\Bigl(Y_y{\hat b}_1+Y_x{\hat b}_1+\frac{S_Y}{n}\Bigr).
\end{displaymath}
Hence, 
\begin{displaymath}
 Z_{21}
\le 
C\frac{S_X^2}{m^3}
\Bigl(
2{\hat b}_1^3\sum_{x\in V}Y_x^3
+
{\hat b}_1^2n^{-1}S_Y^2
\Bigr).
\end{displaymath}
From (\ref{Y3}) and the fact that $n=o_P(S_Y)$
 we conclude that $Z_{21}=o_P(S_Y^2)$.
To estimate $Z_{22}$ and $Z_{23}$ we use the first and second 
identities of (\ref{Z22}). We obtain
\begin{displaymath}
 Q_{(yw)(xw)}
\le 
(Y_x+Y_y){\hat a}_2{\hat b}_1
\qquad
{\text{and}}
\qquad
Q_{(xw)(x\tau)}
\le 
2{\hat b}_1^2X_wX_{\tau}nm^{-1}.
\end{displaymath}
Hence,
\begin{displaymath}
 Z_{22}
\le 
C
{\hat a}_2^3{\hat b}_1^3\sum_{x\in V}Y_x^3,
\qquad
Z_{23}
\le
C
{\hat b}_1^4S_X^2m^{-3}S_Yn.
\end{displaymath}
We note that both quantities on the right are of order $o_P(S_Y^2)$, since
$n=o_P(S_Y)$ by Lemma \ref{L1} and $S_X^2=o_P(m^3)$. We conclude that 
$Z_2=o_P(S_Y^2)$.

Proof of the bound $Z_3=o_P(S_Y^2)$. We split $Z_3=Z_{31}+Z_{32}$, where
\begin{eqnarray}
\nonumber
 Z_{31}
&
=
&
\sum_{y\in V}
\sum_{w\in W}
\sum_{x\in V\setminus\{y\}}
\sum_{\tau\in W\setminus\{w\}}
t_{(yw)(xw)(x\tau)}^2 
\Bigl(
\sum_{z\in V\setminus\{x,y\}}p_{z\tau}
\Bigr)^2,
\\
\nonumber
Z_{32}
&
=
&
\sum_{\{x,y\}\subset V}
\sum_{w\in W}
\sum_{\tau\in W\setminus\{w\}}
\sum_{z\in V\setminus\{x,y\}}
t_{(yw)(xw)(z\tau)}^2\bigl(p_{y\tau}+p_{x\tau}\bigr)^2.
\end{eqnarray}
We have
\begin{displaymath}
 Z_{31}
\le 
C
{\hat a}_2{\hat b}_1^3\frac{S_X}{m^{3/2}}\sqrt{n}S_Y,
\qquad
Z_{32}
\le
C
{\hat a}_2{\hat b}_1^2\frac{S_X}{m^{3/2}}n^{-1/2}\sum_{x\in V}Y_x^3.
\end{displaymath}
By the same argument as above we obtain that $Z_3=o_P(S_Y^2)$.

Finally, we have
\begin{displaymath}
 Z_4\le \sum_{x\in V}
\sum_{y\in V\setminus \{x\}}
\sum_{z\in V\setminus\{x,y\}}
\sum_{w\in W}
\sum_{\tau\in W\setminus\{w\}}
t_{(yw)(xw)(x\tau)(z\tau)}^2
\le
C
{\hat a}_2^2{\hat b}_1^2S_Y=O_P(S_Y)=o_P(S_Y^2).
\end{displaymath}
The proof of the bound 
${\tilde \E}\bigl(N_2-{\tilde \E}N_2\bigr)^2
=
o_P\bigl(S_Y^2\bigr)$ 
is completed.

\medskip

Proof of (ii). Denoting the sum 
$\sum_{\Lambda}\sum_{w\in W}\sum_{\tau\in W\setminus\{w\}}$ by 
$\sum_{*}$
and using the shorthand notation 
\begin{eqnarray}\nonumber
&&
{\mathbb I}^*
=
{\mathbb I}_{xw}{\mathbb I}_{yw}{\mathbb I}_{x\tau}{\mathbb I}_{z\tau}, 
\qquad
p^*= p_{xw}p_{yw}p_{x\tau}p_{z\tau},
\qquad
\lambda^*=\lambda_{xw}\lambda_{yw}\lambda_{x\tau}\lambda_{z\tau},
\\
\nonumber
&&
\delta_1^*=(\lambda_{xw}-p_{xw})\lambda_{yw}\lambda_{x\tau}\lambda_{z\tau},
\qquad
\delta_2^*=p_{xw}(\lambda_{yw}-p_{yw})\lambda_{x\tau}\lambda_{z\tau},
\\
\nonumber
&&
\delta_3^*=p_{xw}p_{yw}(\lambda_{x\tau}-p_{x\tau})\lambda_{z\tau},
\qquad
\delta_4^*=p_{xw}p_{yw}p_{x\tau}(\lambda_{z\tau}-p_{z\tau}),
\end{eqnarray}
we have  $N_2=\sum_*{\mathbb I}^*$ and  ${\tilde \E}{\mathbb I}^*=p^*$,
and ${\tilde \E}N_2=\sum_{*}p^*$. 

We derive  (ii) from the relations shown below
\begin{equation}\label{pl}
 \sum_*p^*=(1+o_P(1)) \sum_{*}\lambda^*
\qquad
{\text{and}}
\qquad
\sum_{*}\lambda^*=2^{-1}{\hat a}_2^2{\hat b}_1^2S_Y+o_P(S_Y).
\end{equation}
To prove the second relation we regroup the sum
\begin{eqnarray}
\nonumber
 \sum_*\lambda^*
 &=&
 \sum_{\Lambda}Y_x^2\frac{Y_yY_z}{n^2}
\Bigl({\hat a}_2^2-\sum_{w\in W}\frac{X_w^4}{m^2}\Bigr)
\\
\label{vasario3b}
&=&
\frac{1}{2}
\sum_{x\in V}
Y_x^2
\left(
\Bigl({\hat b}_1-\frac{Y_x}{n}\Bigr)^2
-
\sum_{z\in V\setminus\{x\}}\frac{Y_z^2}{n^2}
\right)
\Bigl({\hat a}_2^2-\sum_{w\in W}\frac{X_w^4}{m^2}\Bigr).
\end{eqnarray}
For $\E X_1^2<\infty$ and $\E Y_1<\infty$, we obtain from Lemma \ref{L8}
that 
\begin{equation}\nonumber
\sum_{w\in W}\frac{X_w^4}{m^2} =o_P(1),
\quad
\
\sum_{z\in V}\frac{Y_z^2}{n^2}=o_P(1),
\quad
\
\sum_{x\in V}\frac{Y_x^3}{n}=o_P(S_Y),
\quad
\
\sum_{x\in V}\frac{Y_x^4}{n^2}=o_P(S_Y).
\end{equation}
Invoking these bounds in (\ref{vasario3b}) we obtain
the second relation of (\ref{pl}).

To prove the first bound of (\ref{pl}) we write 
\begin{eqnarray}\nonumber
&&
 0\le \sum_{*}\lambda^*-\sum_{*}p_*=\sum_{*}
(\delta_1^*+\delta_2^*+\delta_3^*+\delta_4^*)
\le
\sum_{*}\lambda^*({\mathbb I}_{1}^*+{\mathbb I}_{2}^*+
{\mathbb I}_{3}^*+{\mathbb I}_{4}^*),
\\
\nonumber
&&
{\mathbb I}_{1}^*:={\mathbb I}_{Y_xX_w>\sqrt{nm}},
\quad
{\mathbb I}_{2}^*:={\mathbb I}_{Y_yX_w>\sqrt{nm}},
\quad
{\mathbb I}_{3}^*:={\mathbb I}_{Y_xX_\tau>\sqrt{nm}},
\quad
{\mathbb I}_{4}^*:={\mathbb I}_{Y_zX_\tau>\sqrt{nm}}
\end{eqnarray}
and estimate $\sum_{*}\lambda^*{\mathbb I}_{i}^*=o_P(S_Y)$,
for $i\in[4]$. We only show this bound for $i=1$.
Let $\varepsilon(\cdot)$ be the function
associated with the distribution of $Y_1$ by Lemma \ref{L8}. So that 
$\varepsilon(n)=o(1)$ and 
with probability $1-o_P(1)$ we have $\max_{x\in V}Y_x\le n\varepsilon(n)$. 
If the latter inequality holds, 
then every event ${\mathbb I}_{Y_xX_w>\sqrt{nm}}=1$ implies 
$X_w>\beta_n^{1/2}\varepsilon^{-1}(n)$. We denote the indicator of 
the latter event ${\mathbb I}_w$. We have with probability $1-o(1)$
\begin{displaymath}
 \sum_{*}\lambda^*{\mathbb I}_1^*
\le
\sum_{\Lambda}Y_x^2\frac{Y_yY_z}{n^2}
\sum_{w\in W}\frac{X_w^2}{m}{\mathbb I}_w
\sum_{\tau\in W\setminus\{w\}}\frac{X_\tau^2}{m}
\le
2^{-1}{\hat a}_2{\hat b}_1^2S_Y\sum_{w\in W}\frac{X_w^2}{m}{\mathbb I}_w.
\end{displaymath}
Finally, $\E X_1^2<\infty$ implies 
$\E \sum_{w\in W} \frac{X_w^2}{m}{\mathbb I}_w=o(1)$.
Hence $\sum_{w\in W} \frac{X_w^2}{m}{\mathbb I}_w=o_P(1)$.

\end{proof}


\begin{thebibliography}{}









 
\bibitem{Bloznelis2013}
Bloznelis, M.:
Degree and clustering coefficient in sparse random intersection graphs, 
The Annals of Applied Probability 23 (2013), 1254--1289.




\bibitem{BloznelisDamarackas2013}
Bloznelis, M.,  Damarackas, J.:
Degree distribution of an inhomogeneous random intersection graph.
Electron. J. Comb.  20(3) (2013), $\#$P3. 
 
\bibitem{BloznelisGJKR2015-1}
Bloznelis, M.,  Godehardt, E., Jaworski, J., Kurauskas, V., Rybarczyk, K.:
Recent Progress in Complex Network Analysis - Models 
of Random Intersection Graphs. 
In: Lausen, B., Krolak-Schwerdt, S., B{\"o}hmer, M. (eds) 
Data Science, Learning by Latent Structures, and
Knowledge Discovery,
Springer, Berlin, (2015),
69--78.

 
\bibitem{BloznelisGJKR2015-2}
Bloznelis, M.,  Godehardt, E., Jaworski, J., Kurauskas, V., Rybarczyk, K.:
Recent Progress in Complex Network Analysis
- Properties of Random Intersection Graphs.
In: Lausen, B., Krolak-Schwerdt, S., B{\"o}hmer, M. (eds) 
Data Science, Learning by Latent Structures, and
Knowledge Discovery, 
Springer, Berlin, (2015), 79--88.


\bibitem{BloznelisKurauskas2015}
Bloznelis, M., Kurauskas, K.:
Clustering function: another view on clustering coefficient.
Journal of Complex Networks
4 (2016), 61--86. 


\bibitem{Deijfen}
Deijfen, M., Kets, W.:
    Random intersection graphs with tunable degree distribution and clustering,
Probab. Engrg. Inform. Sci. 23 (2009), 661--674.

%\bibitem{Durret2007}
%Durret, R.: {\it Random Graph Dynamics}, Cambridge University Press, 2007.


\bibitem{FellerII} 
Feller, W.: 
An introduction to probability theory and its applications. Vol. II. 
2nd Edition. 
John Wiley $\&$ Sons, New York, (1971). 


\bibitem{FossKZ2013}
Foss, S.,  Korshunov, D., Zachary, S.:
An Introduction to Heavy-Tailed and Subexponential Distributions. 2nd Edition.  
 Springer, New York, (2013). 

\bibitem{godehardt2003}
   Godehardt, E., Jaworski, J.:
   Two models of random intersection graphs for classification, in:
   Studies in Classification, Data Analysis and Knowledge Organization,
   Springer,
   Berlin,
   (2003),
   67--81.

\bibitem{GJR2012}
  Godehardt, E., Jaworski, J., Rybarczyk, K.:   
  Clustering coefficients of random intersection graphs, in:
   Studies in Classification, Data Analysis and Knowledge Organization,
   Springer,
   Berlin,
   (2012), 243--253.

%\bibitem{Grechnikov2012}
%Grechnikov, E. A.: 
%The Degree Distribution and the Number of Edges Between Vertices
%of Given Degrees in the Buckley-Osthus Model of a Random Web Graph,
%J. of Internet
%Mathematics
%8 (2012), 257--287.



\bibitem{JacobMorters2013}
   Jacob, E., M\"orters, P.:
   A spatial preferential attachment model with local clustering,
   In: Bonato, A., Mitzenmacher, M.,  Pra\l at, P. (eds) 
   Algorithms and Models for the Web Graphs, WAW2013, LNCS 8305,
   Springer,
   Berlin,
   (2013),
   14--25.



   \bibitem{karonski1999}
   Karo\'nski, M., Scheinerman,  E. R., Singer-Cohen, K. B.:
   On random intersection graphs: The subgraph problem,
   Combinatorics, Probability and Computing
   8 (1999),
   131--159.

\bibitem{kurauskas2015}
Kurauskas, V.:
On local weak limit and subgraph counts for sparse random graphs, 
arXiv:1504.08103v2, (2015).


%\bibitem{Mori2005}
%M\'ori, T, F.: The maximum degree of the Barabási-Albert random tree,
%Combinatorics, Probability and Computing 14 (2005), 339--348.

\bibitem{Newman2001}
%Reikalingas
Newman, M. E. J., Strogatz, S. H.,   Watts, D. J.:
Random graphs with arbitrary degree distributions and their applications,
Physical Review E
64 (2001) 
026118.

\bibitem{Newman2001aCUTOFF} 
Newman, M. E. J.:
The structure of scientific collaboration networks, 
Proc. Natl.Acad. Sci. USA 98 (2001), 404--409.

\bibitem{Newman2001bCUTOFF}
Newman, M. E. J.:
Scientific collaboration networks. I and II, 
Physical Review E
64 (2001) 
016131, 016132.


\bibitem{OstroumovaSamosvat2014}
   Ostroumova Prokhorenkova, L., Samosvat, E.:
   Global clustering coefficient in scale-free networks,
   In: Bonato, A., Chung Graham, F.,  Pra\l at, P. (eds) 
   Algorithms and Models for the Web Graphs, WAW2014, LNCS 8882,
   Springer,
   Berlin,
   (2014),
   47--58.

\bibitem{Redner1998}
Redner, S.: How popular is your paper? 
An empirical study of the citation distribution,
Eur.Phys. J. B 4. (1998), 131--134.



\end{thebibliography}
\end{document}